\documentclass[aip,jcp,preprint]{revtex4-1}

\usepackage[utf8]{inputenc}
\usepackage[english]{babel}
\usepackage{amsmath}
\usepackage{mathtools}
\usepackage{amsfonts}
\usepackage{amssymb}
\usepackage{amsthm}
\usepackage{algorithm}
\usepackage{algpseudocode}
\usepackage{graphicx}
\usepackage{hyperref}
\usepackage{silence}
\hypersetup{
    colorlinks=true,
    linkcolor=blue,
    filecolor=magenta,
    urlcolor=cyan,
    citecolor=blue
}
\renewcommand{\thesubsection}{\thesection.\arabic{subsection}}
\makeatletter
\renewcommand{\p@subsection}{} % Add this line
\makeatother

\newtheorem{definition}{Definition}
\newtheorem{lemma}{Lemma}

\newtheorem{theorem}{Theorem}

\begin{document}

\title{Efficient and Scalable Wave Function Compression Using Corner Hierarchical Matrices}

\author{Kenneth O. Berard}
\altaffiliation{Contributed equally}
\altaffiliation{Current address: Department of Chemistry, Brown University, Providence, RI, USA}
\affiliation{Institute for Advanced Computational Science, Stony Brook University, Stony Brook, NY, USA 11794}
\affiliation{Department of Chemistry, Stony Brook University, Stony Brook, NY, USA 11794}

\author{Hongji Gao}
\altaffiliation{Contributed equally}
\affiliation{Institute for Advanced Computational Science, Stony Brook University, Stony Brook, NY, USA 11794}
\affiliation{Department of Applied Mathematics and Statistics, Stony Brook University, Stony Brook, NY, USA 11794}

\author{Alexander Teplukhin}
\altaffiliation{Current Address: Intel Corporation RA2, Hillsboro, OR USA  97124}
\affiliation{Institute for Advanced Computational Science, Stony Brook University, Stony Brook, NY, USA 11794}
\affiliation{Department of Chemistry, Stony Brook University, Stony Brook, NY, USA 11794}

\author{Xiangmin Jiao}
\email[Correspondence email address: ]{xiangmin.jiao@stonybrook.edu}
\affiliation{Institute for Advanced Computational Science, Stony Brook University, Stony Brook, NY, USA 11794}
\affiliation{Department of Applied Mathematics and Statistics, Stony Brook University, Stony Brook, NY, USA 11794}

\author{Benjamin G. Levine}
\email[Correspondence email address: ]{ben.levine@stonybrook.edu}
\affiliation{Institute for Advanced Computational Science, Stony Brook University, Stony Brook, NY, USA 11794}
\affiliation{Department of Chemistry, Stony Brook University, Stony Brook, NY, USA 11794}

\begin{abstract}
The exponential scaling of complete active space (CAS) and full configuration interaction (CI) calculations limits the ability of quantum chemists to simulate the electronic structures of strongly correlated systems. Herein, we present corner hierarchically approximated CI (CHACI), an approach to wave function compression based on corner hierarchical matrices (CH-matrices)---a new variant of hierarchical matrices based on a block-wise low-rank decomposition. By application to dodecacene, a strongly correlated molecule, we demonstrate that CH matrix compression provides superior compression compared to a truncated global singular value decomposition. The compression ratio is shown to improve with increasing active space size. By comparison of several alternative schemes, we demonstrate that superior compression is achieved by a) using a blocking approach that emphasizes the upper-left corner of the CI vector, b) sorting the CI vector prior to compression, and c) optimizing the rank of each block to maximize information density.
\end{abstract}

\maketitle

\section{Introduction}
\label{sec:introduction}

Complete active space (CAS) configuration interaction (CI) and its self-consistent field counterpart (CASSCF) remain the quantum chemist's workhorse method for describing multi-reference systems.\cite{Roos1980} The most frequent targets for multi-reference calculations include strongly correlated systems, dissociated molecules, transition states, conical intersections involving the electronic ground state, and excited electronic states with multiple excited characters. CAS methods are viable when the multi-reference effects in question can be described in an active orbital space of fewer than $\sim$20 electrons. Unfortunately, the factorial scaling of the CAS wave function with active space size renders it impractical for anything larger.

As such, strategies for using compressed representations of the CASCI wave function have been widely explored in recent years, exploiting the large number of zero or near-zero elements of the CI coefficient vector (``configurational deadwood''\cite{Bytautas2009}) to reduce compute time and storage. For example, restricted, generalized, and other active space methods use different truncations of the CAS expansion to reduce storage and cost.\cite{Malmqvist1990,Fleig2001,Fleig2003,MaLiManni2011,Casanova2022,Parker2013,Ivanic2003} Similarly, selected CI approaches use various physical criteria to develop a sparse representation of the CI vector on-the-fly.\cite{Huron1973,Harrison1991,Daudey1993,Greer1995a,Greer1995b,Coe2014a,Coe2014b,Evangelista2014,Coe2018,Holmes2016,Schriber2016,Ohtsuka2017,Tubman2020,Goings2021,Neese2003,Nakatsuji2005,Abrams2005,Liu2016,Holmes2017,Li2018,Abraham2020} Seniority CI achieves an efficient sparse representation of strongly correlated systems via a truncation of the CI space based on seniority number.\cite{Bytautas2011} Full configuration interaction quantum Monte Carlo (FCI-QMC) methods can describe very large active spaces without explicit representation of the entire CI vector, allowing a fixed number of walkers to explore the most important regions of the CI vector stochastically.\cite{Booth2009,Shepherd2014,Li2019} An algorithm for choosing orbitals to maximize CI vector sparsity in exchange-coupled systems has also recently been reported.\cite{LiManniDobrautz2020}

Other methods take advantage of data sparsity rather than sparsity. Whereas a sparse matrix has many zero elements, a data-sparse matrix may contain few zero elements but is compressible because it is low-rank or has other redundancies that may be exploited. For example, a highly efficient representation of the strongly correlated wave function of one-dimensional and quasi-one-dimensional systems can be achieved via the density matrix renormalization group (DMRG) approach.\cite{White1992,Chan2011,Verstraete2023} The DMRG method can be thought of as leveraging the data sparsity of the CI vector. Rank-reduced configuration interaction methods also leverage data sparsity through a low-rank approximation of the CI vector,\cite{Olsen1987,Lindh1988,Koch1992,Taylor2013,Fales2018} and in a similar spirit, ideas from compressive sensing have been adapted to wave function compression.\cite{Knowles2015} Going one step further, it has been argued that the wave function itself is not an essential descriptor of electronic structure and can be replaced by the electron density\cite{Kohn1965} or reduced density matrices.\cite{Coleman1963,Sharma2008,Mazziotti2012} The search for an efficient and broadly applicable compressed wave function representation remains an active and important research direction, and recently there have been efforts to systematize the benchmarking of methods for compression of strongly correlated wave functions.\cite{Stair2020,Eriksen2021}

In this paper, we present a proof-of-concept demonstration of a novel wave function compression strategy based on corner hierarchical matrices (CH-matrices), a new variant of hierarchical matrices. CH-matrices are motivated by H-matrices and H$^2$-matrices,\cite{Boerm2003,Lin2011,Hackbusch2015} which use a block-wise low-rank approximation to a non-sparse matrix, leveraging data sparsity to minimize both storage and the computational cost of matrix operations. For \(N \times N\) matrices, the H-matrix representation allows storage and matrix multiplication with quasi-linear scaling (\(\mathcal{O}(N\log N)\) and \(\mathcal{O}(N\log^2 N)\), respectively). Open-source software is available to manipulate various flavors of hierarchical matrices, and given their block-wise nature, hierarchical matrices are well-suited for massively parallel implementation.

The effectiveness of H- and H$^2$-matrices relies on some degree of diagonal dominance of the matrices. This is because the matrix is subdivided into blocks, with smaller blocks on the diagonal and increasing block size as you move away from the diagonal, as illustrated in Figure~\ref{fig:hierarchical-matrix}. The blocks closest to the diagonal, which are shaded gray, are represented as dense matrices with no loss of accuracy. The off-diagonal blocks (shown in white) are approximated as low-rank. Typically, the rank is held constant despite the fact that the blocks are of different sizes, resulting in greater compression of the data farthest from the diagonal. These hierarchical-matrix techniques are beginning to see use in quantum chemistry. Chow and coworkers have leveraged hierarchical matrices to compress the electron repulsion integrals (ERI).\cite{Xing2020a,Xing2020b,Huang2021} In this context, the hierarchical-matrix approach can be thought of as an algebraic generalization of the fast multipole method.\cite{Barnes1986,Greengard1987,Strain1996} Using the proxy point method, the hierarchically compressed ERI tensor may be computed for an arbitrary contracted Gaussian basis, in a similar spirit. Unfortunately, the CASCI wave function is not diagonally dominant, so these hierarchical-matrix approaches are not very effective in our experience. Instead, we observe that the CASCI wave function is dominated by the upper-left corner of the CI vector. Based on this observation, we have developed a new variant of hierarchical matrices as illustrated in Figure~\ref{fig:corner-hierarchical}, which we call Corner Hierarchical matrices (CH-matrices), to compress matrices with such patterns.

\begin{figure}[htbp]
    \centering
    \begin{minipage}[t]{0.48\textwidth}
        \centering
        \includegraphics[width=0.94\columnwidth]{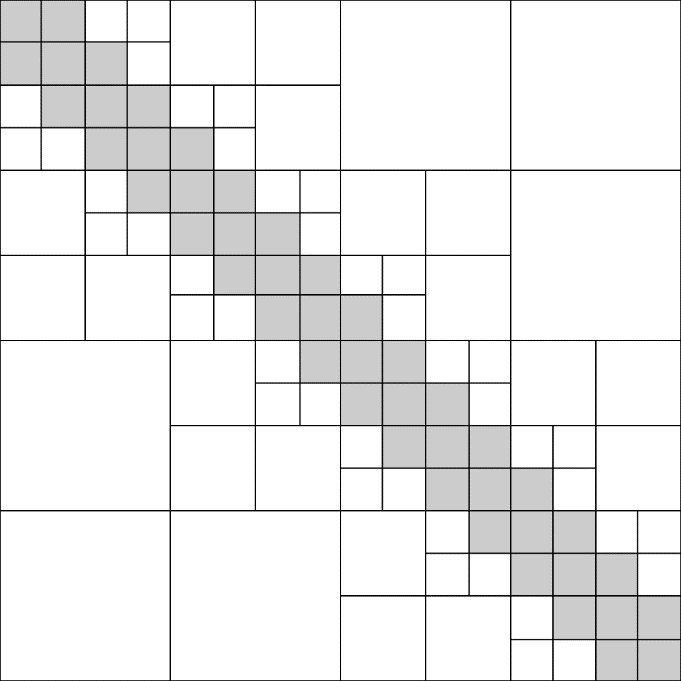}
        \caption{Block representation of an H-matrix. White blocks are stored as low-rank approximations, while gray blocks are stored as dense matrices.}
        \label{fig:hierarchical-matrix}
    \end{minipage}
    \hfill
    \begin{minipage}[t]{0.48\textwidth}
        \centering
        \includegraphics[width=0.94\columnwidth]{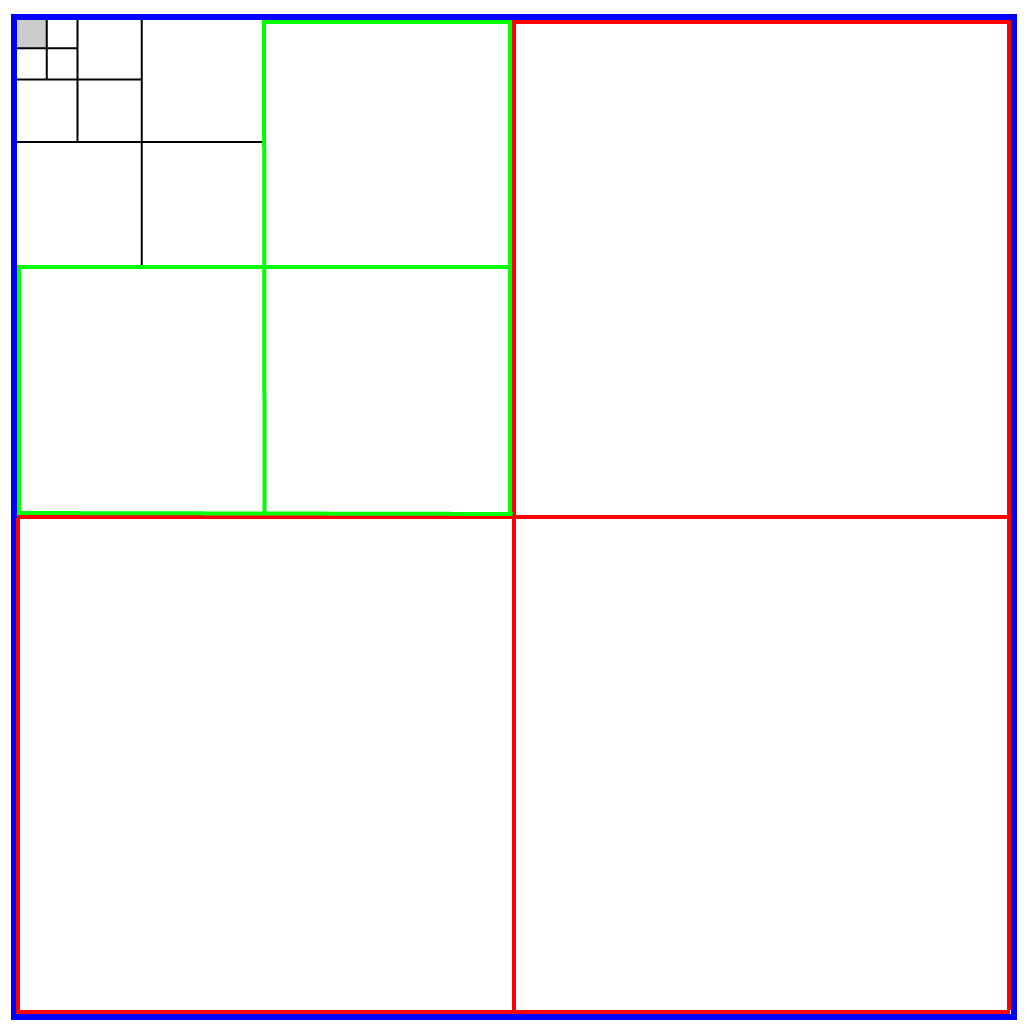}
        \caption{Schematic of the corner-hierarchical matrix structure. The upper-left corner is compressed hierarchically and the other blocks are compressed adaptively.}
        \label{fig:corner-hierarchical}
    \end{minipage}
\end{figure}

% \begin{figure}[htbp]
%     \centering
%     \includegraphics[width=3.06in, height=3.06in]{figures/image1}
%     \caption{Block representation of an H-matrix. White blocks are stored as low-rank approximations, while gray blocks are stored as dense matrices.}
%     \label{fig:hierarchical-matrix}
% \end{figure}
% \begin{figure}[htbp]
%     \centering
%     \includegraphics[width=3.06in, height=3.06in]{figures/image2a.png}
%     \caption{ The CH matrix structure.}
%     \label{fig:corner-hierarchical}
% \end{figure}

Herein, we introduce our CH-matrix approach to compress the CASCI vectors of strongly correlated systems. In section~\ref{sec:theory-and-methods}, we define the corner hierarchically approximated CI (CHACI) representation of the wave function, along with several other candidate compression strategies. In section~\ref{sec:results-and-discussion}, we compare the performance of CHACI to several other approaches, including a truncated singular value decomposition approach that represents an optimal global low-rank approximation, akin to rank-reduced full CI (RRFCI). Finally, in section~\ref{sec:conclusions}, we present conclusions and discuss the prospects for the direct optimization of CHACI wave functions. Some details and analysis of the CHACI algorithm are presented in the appendix.

\section{Theory and Methods}
\label{sec:theory-and-methods}
Ultimately we seek to directly optimize the electronic wave function in a compressed format. However, in the present paper, we do not solve for the wave function in a compressed format. The purpose of this paper is to demonstrate that hierarchical matrices provide a viable strategy for wave function compression. To this end, we will solve for the full CASCI wave function via traditional approaches, and then compress the resulting wave function into one of several forms. We can then decompress, and apply several metrics to quantify the accuracy of the compressed wave function.

Throughout this work, we will define the set of Slater determinants, \(\left\{\left| \alpha \beta \right\rangle\right\}\), in the CASCI wave function in terms of their composite \(\alpha\) and \(\beta\) strings. The \(\alpha\) (\(\beta\)) strings index the possible ways that \(N_{\alpha}\) (\(N_{\beta}\)) spin-up (spin-down) electrons may occupy \(N_{\text{MO}}\) active spatial molecular orbitals. Therefore, the CASCI wave function may be represented by
\begin{equation}
\left| C \right\rangle = \sum_{\alpha \beta} C_{\alpha \beta} \left| \alpha \beta \right\rangle.
\end{equation}
The set of expansion coefficients, \(\{C_{\alpha \beta}\}\), are typically known as the ``CI vector.''   As such, it is typically thought of as an \(M_{\alpha} M_{\beta}\) vector, but in this work, we will treat it as an \(M_{\alpha} \times M_{\beta}\) matrix, \(\mathbf{C}\). Here \(M_{\alpha}\) and \(M_{\beta}\) represent the number of \(\alpha\) and \(\beta\) strings, respectively,
\begin{equation}
M_{\alpha} = \frac{N_{\text{MO}}!}{N_{\alpha}!(N_{\text{MO}} - N_{\alpha})!}
\end{equation}
and
\begin{equation}
M_{\beta} = \frac{N_{\text{MO}}!}{N_{\beta}!(N_{\text{MO}} - N_{\beta})!}.
\end{equation}
Where not otherwise noted, the strings will be ordered according to the scheme of Duch.\cite{Duch1986}

Below, we define several strategies for wave function compression. CHACI compression, described in subsection~\ref{subsec:haci-wave-function-compression}, is the most efficient strategy we discovered in this work. Several other compression schemes, which we present in order to analyze the necessity of the various features of the CHACI algorithm, are also defined in subsection~\ref{subsec:haci-wave-function-compression}. In subsections~\ref{subsec:model-problems-and-computational-details} we describe the model problems that we choose to test our algorithm and present computational details. In subsection~\ref{subsec:performance-metrics}, the performance metrics that we use to quantify the accuracy of the compressed wave function are defined.

\subsection{CHACI Wave Function Compression}
\label{subsec:haci-wave-function-compression}

The CHACI wave function compression scheme is based on what we call a corner hierarchical (CH) compression. In contrast to the H-matrix scheme illustrated in Figure~\ref{fig:hierarchical-matrix}, which was conceived to compress diagonally-dominant matrices, the CH scheme (Figure~\ref{fig:corner-hierarchical}) is designed to compress matrices that are dominated by the upper-left corner of the matrix. Given the CI vector stored as an $M\times M$ matrix, the CH compression algorithm proceeds as follows. The total matrix (blue) is subdivided into four approximately equal-sized subblocks. The upper-right, lower-left, and lower-right subblocks (red in Figure~\ref{fig:corner-hierarchical}) are compressed adaptively using either a low-rank or dense approximation. The upper left subblock is further subdivided, yielding four smaller subblocks. In a similar fashion, three (green) subblocks are again compressed. The number of levels is approximately logarithmic in $M$ and determined \textit{a priori}. Given the number of levels, we then determine whether to store each block as sparse or dense, as we describe in Algorithm~\ref{alg:optimal-chaci} and analyze its near optimality in Appendix~\ref{sec:near-optimal-chaci-compression}.
%{\color{red}The upper left subblock is recursively subdivided again and again until eventually low-rank compression no longer results in a reduction of storage compared to dense storage. At this point, the remaining subblock is stored densely (gray-shaded block), and the compression is complete.} The pseudocode for a recursive subroutine that would perform the compression into the CHACI form is given in Figure~\ref{fig:corner-hierarchical}b.

Throughout this work, we will employ truncated singular value decomposition (TSVD) to compress individual blocks. The SVD of an \(m \times n\) block, \(\mathbf{B}\), is
\begin{equation}
\mathbf{B} = \mathbf{U}\mathbf{\Sigma} \mathbf{V}^{\dagger},
\end{equation}
where \(\mathbf{U}\) and \(\mathbf{V}\) are, respectively, \(m \times m\) and \(n \times n\) unitary matrices comprising the left and right singular vectors, and \(\mathbf{\Sigma}\) is a diagonal matrix, with diagonal elements equal to the singular values. In the context of this work, $m$ and $n$ are typically equal, but they may differ in general for rectangular matrices. The SVD itself is an exact representation of the matrix \(\mathbf{B}\), and does not provide any compression. However, truncation by elimination of the smallest singular values and corresponding singular vectors provides an efficient low-rank approximation to \(\mathbf{B}\),
\begin{equation}
\mathbf{B} \approx \tilde{\mathbf{B}} = a\tilde{\mathbf{B}}' = a\mathbf{U}_T \mathbf{\Sigma}_T \mathbf{V}_T^{\dagger},
\end{equation}
where \(a\) is a scalar normalization factor, defined below. Upon compression \(\mathbf{U}_T\) and \(\mathbf{V}_T^{\dagger}\) are now \(m \times k\) and \(k \times n\), where \(k\) is the compression rank. Efficient compression can be achieved when an accurate representation of \(\mathbf{B}\) is achieved for \(k \ll \min(n, m)\). TSVD compression decreases the Frobenius norm of the block, and therefore the overall wave function. To compensate for this, we rescale each block to maintain its original norm after compression, according to
\begin{equation}
    \label{eq:rescale}
a = \frac{\|\mathbf{B}\|_F}{\|\mathbf{U}_T \mathbf{\Sigma}_T \mathbf{V}_T^{\dagger}\|_F}.
\end{equation}
This block-wise renormalization prevents the shifting of population from sparser to denser blocks, which would be an artifact of global renormalization. Note that truncated SVD, without subsequent renormalization, is an optimal low-rank approximation to the block, in that the Frobenius norm of the difference between \(\tilde{\mathbf{B}}'\) and \(\mathbf{B}\),
\begin{equation}
\|\tilde{\mathbf{B}}' - \mathbf{B}\|_F = \sum_{i, j}(\tilde{B}'_{ij} - B_{ij})^2,
\end{equation}
is minimized, for a given rank. However, this condition does not guarantee variational optimality (minimal energy).

In CHACI, CH compression of \(\mathbf{C}\) is supplemented by two additional strategies aimed at allowing more efficient compression. First, prior to compression, the rows (columns) of \(\mathbf{C}\) are sorted such that the rows (columns) are arranged in descending order according to their \(L^2\) norms. In this way, the large elements are concentrated in the upper left corner of \(\mathbf{C}\), to the extent possible. 

Secondly, in the CHACI scheme, we do not employ the same value of the compression rank for all blocks. Instead, the SVD of each block is truncated to a different rank, such that the overall Frobenius norm of the compressed wave function, \(\|\tilde{\mathbf{C}}\|_F\), is optimized.  A detailed derivation and algorithm are presented in the appendix, but here we summarize the primary features.
%(Note that the pseudocode in the appendix assumes that the input matrix, \(\mathbf{C}\), has already been sorted.)
The information density, \(\rho_i\), associated with each singular vector pair is computed according to
\begin{equation}
\rho_i = \frac{\sigma_i^2}{n_\text{row}+n_\text{col}+1},
\end{equation}
where \(\sigma_i\) is the singular value, and $n_\text{row}$ and $n_\text{col}$ are the lengths of the left and right singular vectors, respectively. Typically, $n_\text{row}$ and $n_\text{col}$ are equal in our algorithms, but they may differ in general for rectangular blocks. The information density can be thought of as the total contribution to \(\|\tilde{\mathbf{C}}\|_F\) per unit of storage.

Only the singular vector pair whose information density is above a user-defined threshold will be stored.  At a given storage value, this algorithm approximately maximizes the Frobenius norm of the wave function, \(\|\tilde{\mathbf{C}}\|_F\).  Equivalently, this can be thought of as discarding the least amount of information.  However, this algorithm does not necessarily optimize the energy. However, when higher accuracy is required and the threshold for the Frobenius norm error is set very low, several of the upper-left corner blocks may become so dense that the storage cost of TSVD exceeds that of the dense format. To address this, we add another test after the TSVD: If the memory cost of TSVD is no less than the dense format, we directly store the block in dense format. For completeness, we present the pseudocode and its detailed analysis in the Appendix.

\begin{table}[htbp]
\centering
\caption{Summary of the features of the wave function compression schemes compared in this work.}
\begin{tabular}{l|l|c|c} \hline\hline
Scheme & Blocking & Sorting? & Optimal Rank? \\ \hline
CHACI & Corner Hierarchical & Y & Y \\
SR-CHACI & Corner Hierarchical & Y & N \\
U-CHACI & Corner Hierarchical & N & Y \\
H-Matrix & Diagonally Dominant & N & N \\
Truncated SVD & None & N & N \\ \hline\hline
\end{tabular}
\label{tab:compression-schemes}
\end{table}

Taken together, we refer to CH compression of the reordered \(\mathbf{C}\) matrix with optimal rank as CHACI compression, without any additional modifier. To analyze the necessity of different features of the CHACI algorithm, we will present results for several other compression schemes, summarized in Table~\ref{tab:compression-schemes}. The static rank (SR-) CHACI scheme is identical to CHACI, except that the rank of the compressed blocks is held constant rather than optimized per block. We introduce SR-CHACI in order to quantify the utility of dynamically increasing the rank. The unsorted (U-) CHACI scheme is identical to CHACI, except that the rows/columns remain ordered according to the original Duch scheme, rather than sorted by norm. This scheme is introduced to quantify the impact of sorting.

In addition, we will test H-matrix compression\cite{Hackbusch1999} of the wave function. The H-matrix representation is an existing hierarchical representation that is not based on CH blocking. Instead, the blocking structure of the matrix is optimized to represent diagonally dominant matrices, as pictured in Figure~\ref{fig:hierarchical-matrix}. In our implementation, the sparse blocks are compressed by the same TSVD approach that we used for CHACI. The rows/columns are not sorted, and the rank is held constant for all blocks.

Finally, as an important point of comparison, we include a non-hierarchical compression scheme: TSVD of the entire \(\mathbf{C}\) matrix. This form of the wave function is akin to that in RRFCI.\cite{Olsen1987,Lindh1988,Koch1992,Taylor2013,Fales2018} To maintain normalization, we rescale as in Eq.~\eqref{eq:rescale}. Given that implementation of a CI solver based on a hierarchically compressed wave function will be challenging, CHACI must provide a significant improvement in performance over TSVD to warrant further consideration.

\subsection{Model Problems and Computational Details}
\label{subsec:model-problems-and-computational-details}
\begin{figure}[htbp]
    \centering
    \includegraphics[width=4.99in, height=1.22in, trim=1.50in 1.50in 0.00in 1.40in]{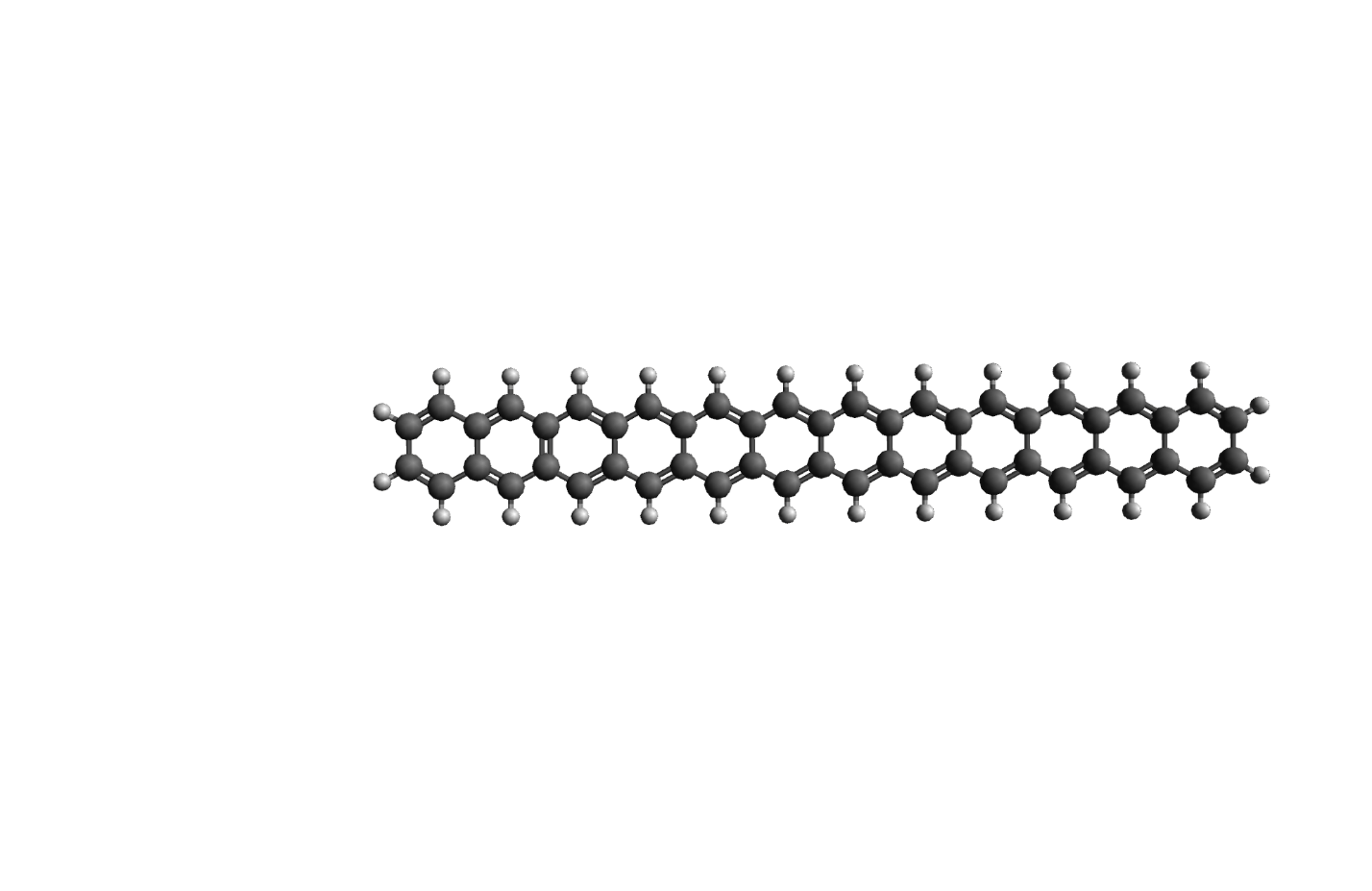}
    \caption{Molecular structure of dodecacene (12-acene).}
    \label{fig:dodecacene-structure}
\end{figure}

\begin{table}[htbp]
\centering
\caption{Total number of double precision floating point values stored in a dense (uncompressed) representation of the singlet wave function CI vector as a function of active space.  Data is given in kdoubles (thousands of doubles).}
\begin{tabular}{c|c} \hline\hline
Active Space & Dense CI Vector Storage (kdoubles) \\ \hline
10-10 & 63 \\ \hline
12-12 & 854 \\ \hline
14-14 & 11,779 \\ \hline
16-16 & 165,637 \\ \hline
\end{tabular}
\label{tab:dense-storage}
\end{table}

As a test case, we have chosen to compute the electronic structure of dodecacene (12-acene; Figure~\ref{fig:dodecacene-structure}). Naively, one might think that longer acenes like 12-acene would have a simple closed-shell, aromatic electronic structure, but they are actually strongly correlated. These systems have significant poly-radical character, so much so that dodecacene is unstable under ambient conditions and has only recently been synthesized under ultrahigh vacuum.\cite{Eisenhut2020} The singlet ground state geometry of 12-acene was optimized at the CAM-B3LYP\cite{Yanai2004} level of theory. Floating occupation molecular orbital\cite{Granucci2000,Slavicek2010} (FOMO) CASCI calculations were then performed with
active spaces of 10-10, 12-12, 14-14, and 16-16, where active spaces are abbreviated \(\langle\text{number of active electrons}\rangle\)-\(\langle\text{number of active orbitals}\rangle\). Table~\ref{tab:dense-storage} presents the total storage required for the CI vector of the exact singlet wave function for each active space, measured in thousands of double-precision floating point values (kdoubles). The storage requirement for triplet wave functions is less by a factor of 0.7-0.8. The STO-3G minimal basis was used for all calculations. In all cases, both the singlet and triplet ground state wave functions are computed at the singlet ground state geometry. All calculations were performed in the TeraChem electronic structure software package.\cite{Ufimtsev2009,Seritan2021,Fales2015,Hohenstein2015}

\subsection{Performance Metrics}
\label{subsec:performance-metrics}
We use several metrics to evaluate the performance of our various compression schemes. All are derived by comparing the properties of the compressed wave function, \(\left| \tilde{C} \right\rangle\), with those of the uncompressed wave function, \(\left| C \right\rangle\). Hereafter we will refer to the latter as the ``exact'' wave function, understanding that it is only exact within the FOMO-CASCI approximation, with a given active space and basis. In all cases, the properties of the compressed wave function are generated by a) outputting the CI vector of the uncompressed wave function from TeraChem, b) applying one of the lossy compression schemes described above, implemented in an external code, c) decompressing the wave function back to its full dimensionality, and d) feeding the resulting CI vector back into TeraChem to compute the relevant property.

To evaluate how well our compressed wave function reproduces relative energies, we compute the vertical singlet-triplet gaps of the exact and compressed wave functions, respectively
\begin{equation}
\Delta E_{S-T} = \langle C_T | \hat{H} | C_T \rangle - \langle C_S | \hat{H} | C_S \rangle
\end{equation}
and
\begin{equation}
\Delta \tilde{E}_{S-T} = \langle \tilde{C}_T | \hat{H} | \tilde{C}_T \rangle - \langle \tilde{C}_S | \hat{H} | \tilde{C}_S \rangle.
\end{equation}
Here the subscripted \(T\) and \(S\) denote the triplet and singlet wave functions, respectively. We report the absolute difference between these values as the error in the gap,
\begin{equation}
\Delta \Delta E_{S-T} = |\Delta \tilde{E}_{S-T} - \Delta E_{S-T}|.
\end{equation}

To quantify the accuracy of the absolute energies, we also report signed errors in the singlet and triplet energies, respectively
\begin{equation}
\Delta E_{S} = \langle \tilde{C}_S | \hat{H} | \tilde{C}_S \rangle - \langle C_S | \hat{H} | C_S \rangle
\end{equation}
and
\begin{equation}
\Delta E_{T} = \langle \tilde{C}_T | \hat{H} | \tilde{C}_T \rangle - \langle C_T | \hat{H} | C_T \rangle.
\end{equation}
Per the variational principle, these errors are non-negative. Note that throughout this work, we compare the energies of the compressed wave functions to the energies of the exact wave functions using the same active space. Thus, our computed errors only include those associated with the compression algorithm, not with the size of the active space.

Lastly, as a separate measure of the accuracy of the wave function, we compute the absolute error in the total spin angular momentum squared (spin contamination) of the compressed singlet and triplet wave functions, respectively
\begin{equation}
\Delta \langle S^2 \rangle = \langle \tilde{C}_S | \hat{S}^2 | \tilde{C}_S \rangle
\end{equation}
and
\begin{equation}
\Delta \langle S^2 \rangle = |\langle \tilde{C}_T | \hat{S}^2 | \tilde{C}_T \rangle - 2|.
\end{equation}
The total spin expectation value is computed using the direct algorithm described in Ref. \citenum{Fales2017}.

In many of the graphs below, the data is plotted as a function of the total number of double-precision variables stored in the compressed representation. The total storage is varied indirectly through the user-chosen parameters of the compression, as described in subsection \ref{subsec:haci-wave-function-compression} and in the Appendix. This provides an apples-to-apples comparison of the accuracy of the wave function at a given level of compression across different compression schemes.

\section{Results and Discussion}
\label{sec:results-and-discussion}
Below, we analyze the performance of CHACI. In subsection~\ref{subsec:accuracy-of-haci-compression-versus-svd}, we compare the accuracy of CHACI compression to that of TSVD compression, as a function of the total storage. In subsection~\ref{subsec:extrapolation-of-performance}, we analyze the prospects for extending CHACI to larger active spaces. In subsections~\ref{subsec:effect-of-dynamic-rank-on-compression} and \ref{subsec:effect-of-sorting-on-compression}, we analyze the necessity of optimal rank and sorting, respectively. Finally, in subsection~\ref{subsec:effect-of-upper-quadrant}, we compare H-matrix compression to TSVD compression.

\subsection{Accuracy of CHACI Compression Versus TSVD}
\label{subsec:accuracy-of-haci-compression-versus-svd}
\begin{figure}[htbp]
    \centering
    \includegraphics[width=6.41in, height=4.49in]{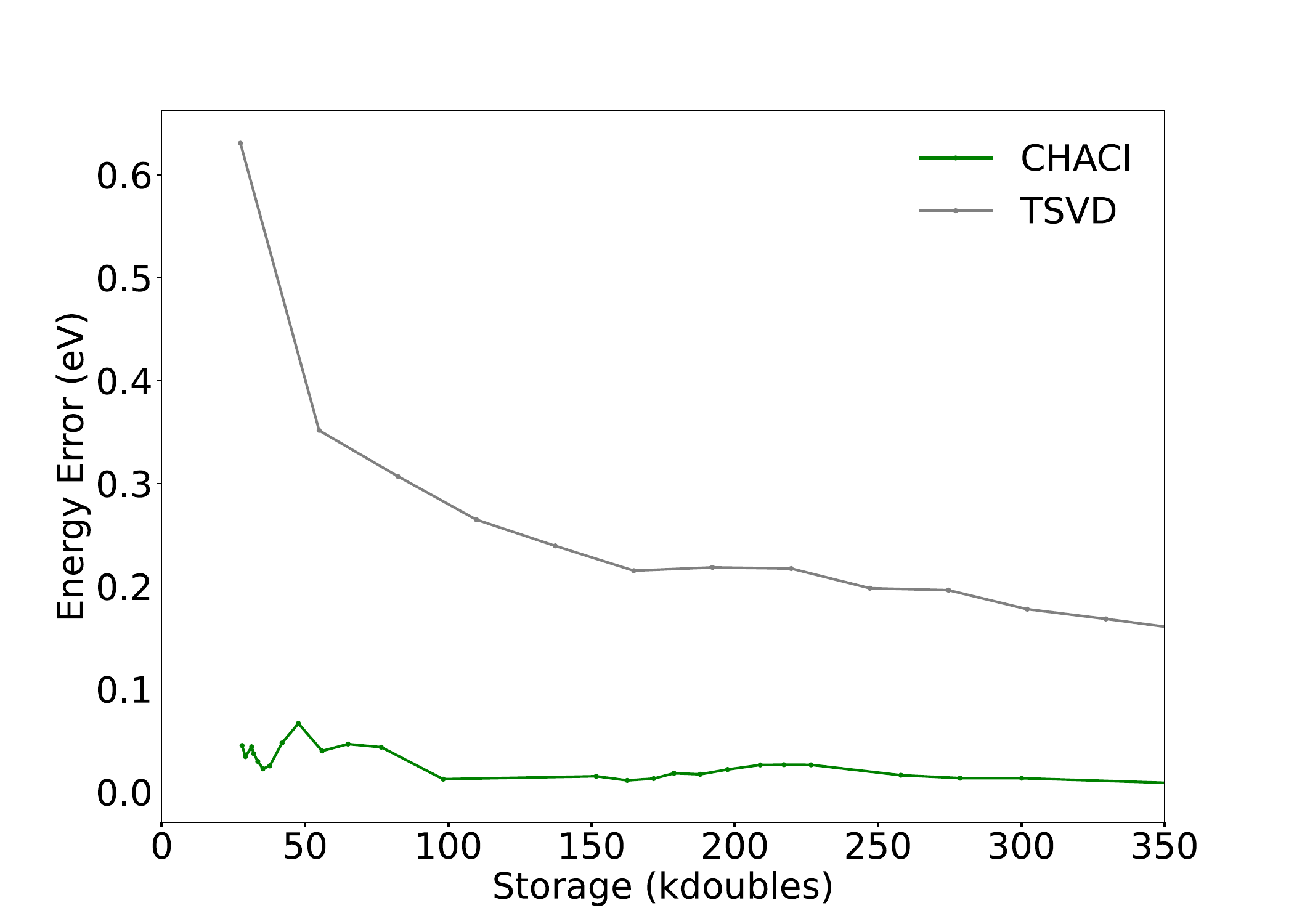}
    \caption{The error in the singlet-triplet gap of 12-acene as a function of total storage, computed with a 14-14 active space. The gray and green lines represent the error corresponding to the compression of the wave function using TSVD and CHACI, respectively.}
    \label{fig:error-gap-14-14}
\end{figure}

We start by comparing the performance of CHACI compression to TSVD compression for the 14-14 active space. Figure~\ref{fig:error-gap-14-14} presents the error in the singlet-triplet gap as a function of total storage. Clearly, CHACI outperforms TSVD in this case. The CHACI error is always less than 0.07 eV, even when only 28 kdoubles are stored (compared to 11,778 kdoubles for dense storage). Truncated SVD also achieves substantial compression, but errors of 0.2 eV remain even with 220 kdoubles stored.

\begin{figure}[htbp]
    \centering
    \includegraphics[width=5.42in, height=5.42in]{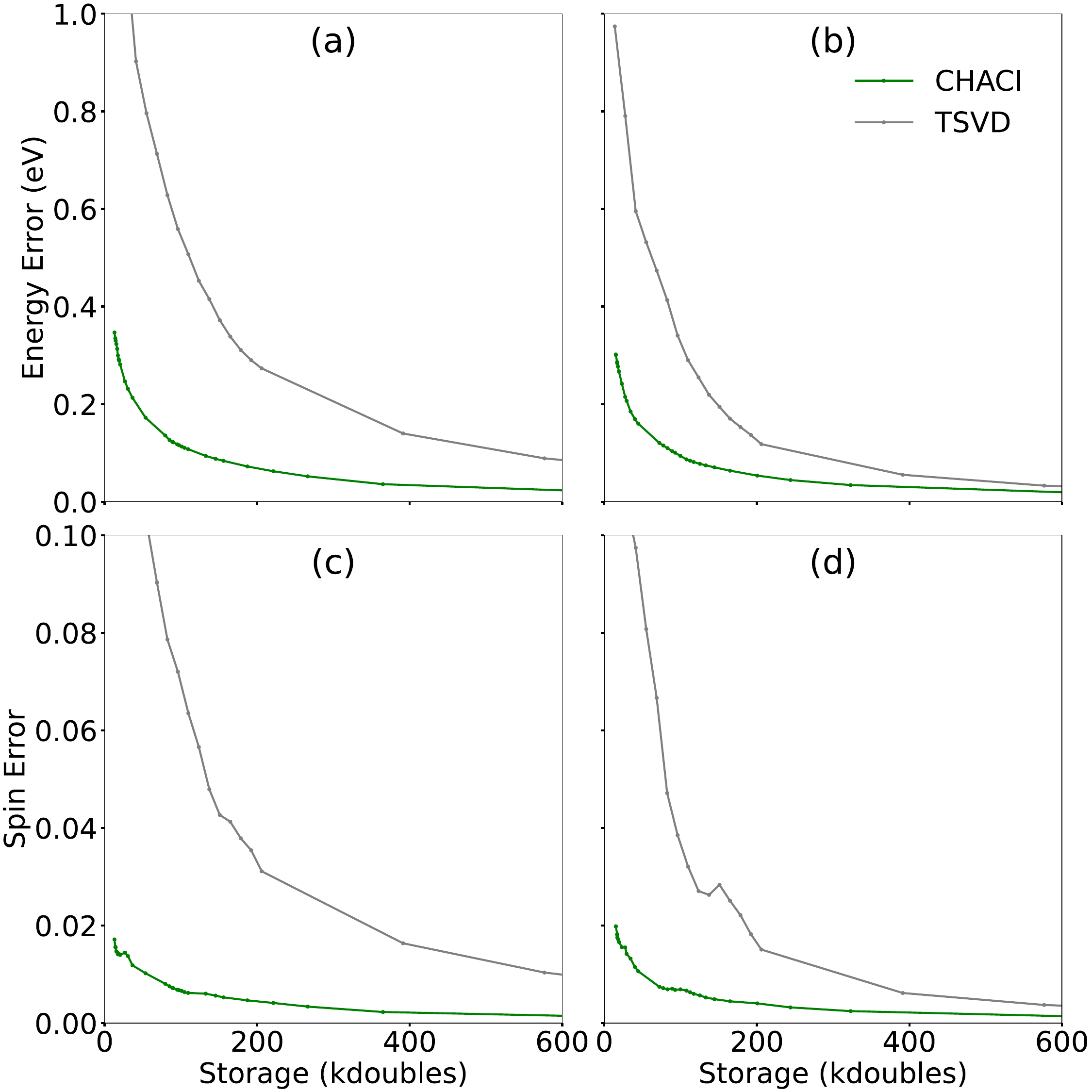}
    \caption{Absolute energy (a and b) and spin (c and d) errors as a function of the storage for 12-acene with a 14-14 active space. Panel (a) and (c) correspond to the singlet wave function, while (b) and (d) correspond to the triplet wave function. The gray and green lines correspond to TSVD and CHACI compression, respectively.}
    \label{fig:energy-spin-errors-14-14}
\end{figure}

Figures~\ref{fig:energy-spin-errors-14-14}a and b present errors in the singlet and triplet absolute energies, respectively. Again, the performance of CHACI is far superior to TSVD, especially in the small-storage regime, where the TSVD errors are a substantial fraction of an eV. Figure~\ref{fig:energy-spin-errors-14-14}c and d present spin error. Here CHACI again outperforms TSVD. Even with very modest (28 kdoubles) storage, the error in \(\langle S^{2} \rangle\) is 0.02, and the error drops rapidly toward zero with additional storage. In contrast, TSVD spin errors for the singlet case do not drop below 0.02 until nearly 400 kdoubles are stored.

\begin{figure}[htbp]
    \centering
    \includegraphics[width=6.47in, height=4.67in]{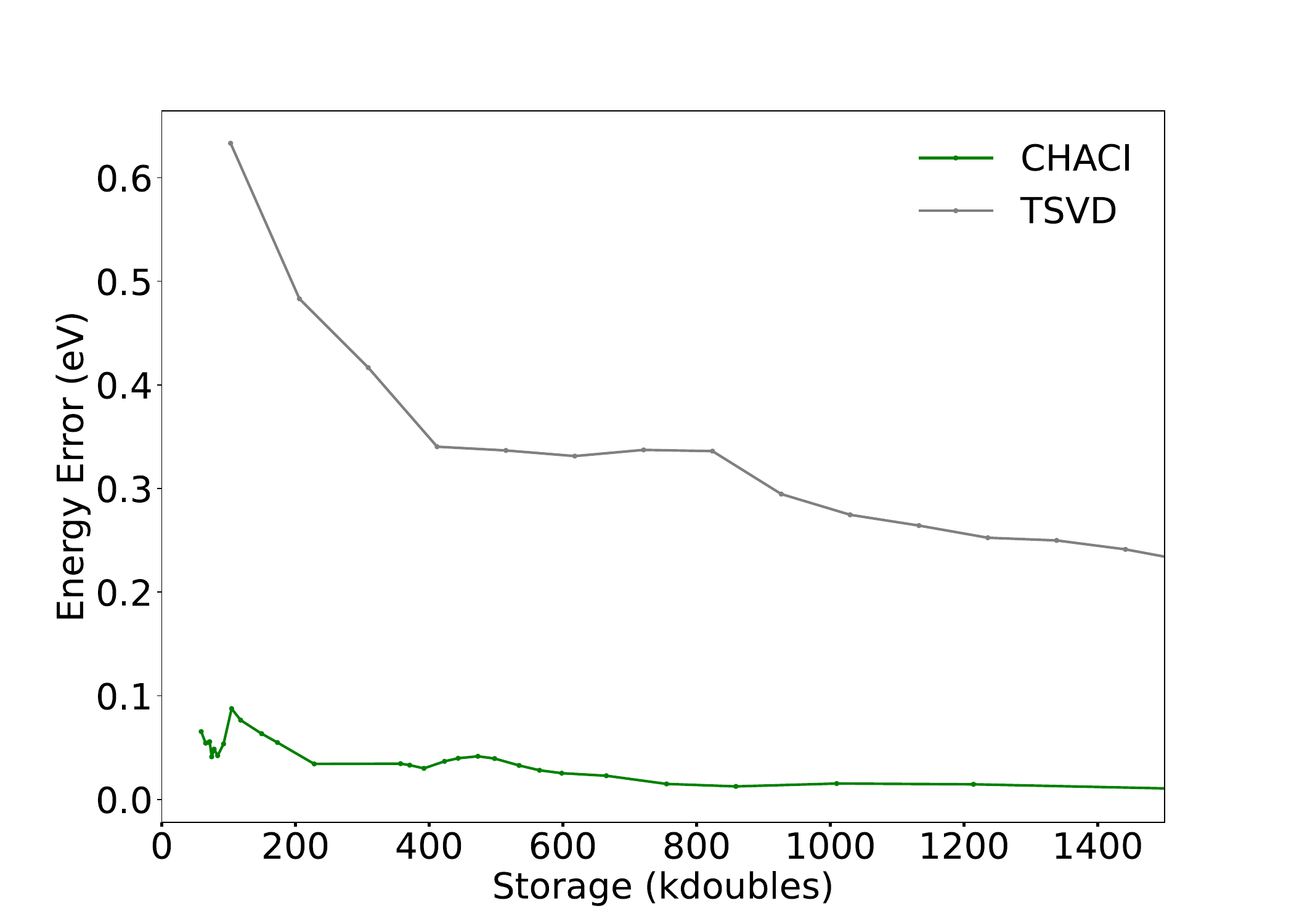}
    \caption{The error in the singlet-triplet gap of 12-acene as a function of total storage, computed with a 16-16 active space. The gray and green lines represent the error incurred by compression of the wave function using TSVD and CHACI, respectively.}
    \label{fig:error-gap-16-16}
\end{figure}

\begin{figure}[htbp]
    \centering
    \includegraphics[width=6.50in, height=6.50in]{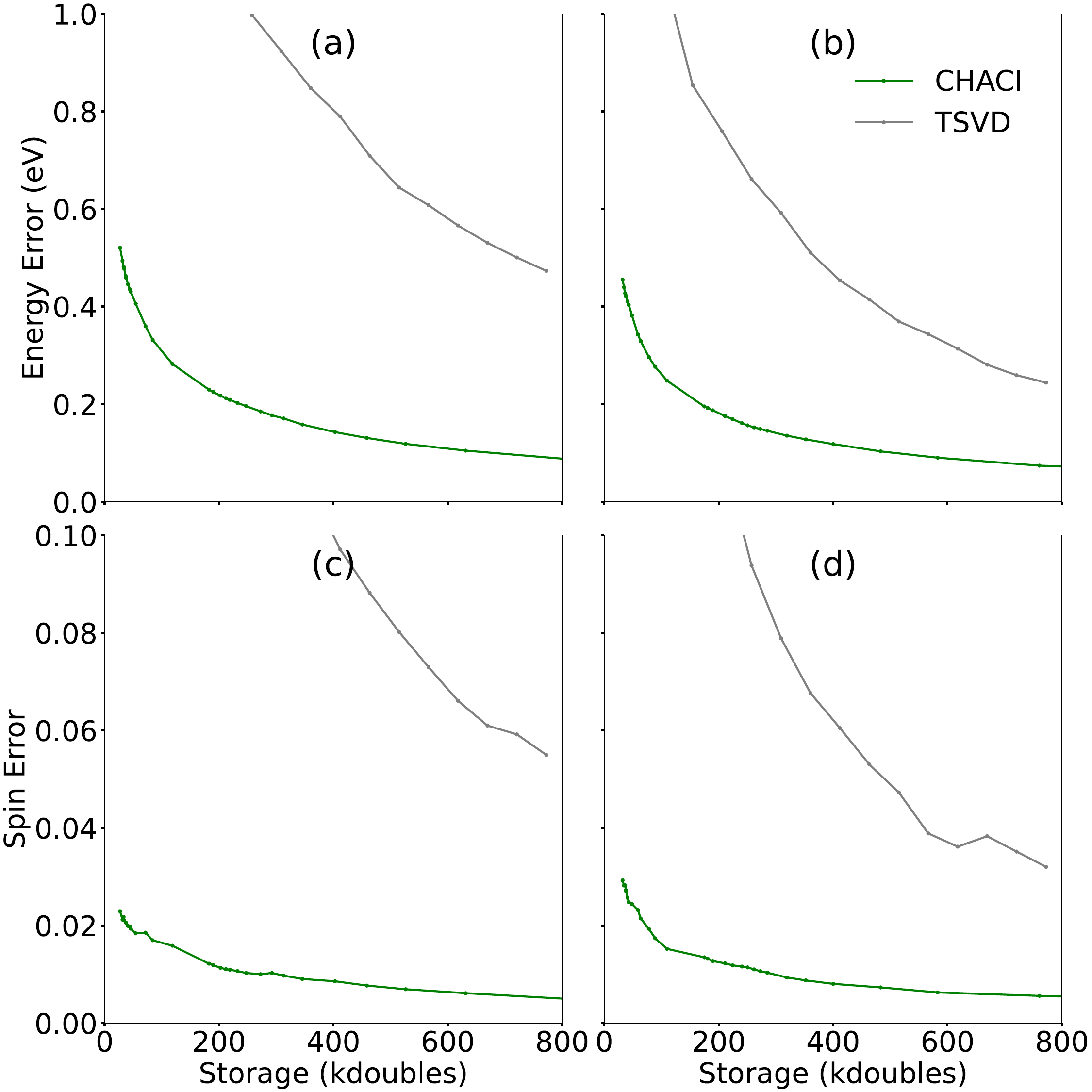}
    \caption{Absolute energy (a and b) and spin (c and d) errors as a function of the storage for 12-acene with a 16-16 active space. Panel (a) and (c) correspond to the singlet wave function, while (b) and (d) correspond to the triplet wave function. The gray and green lines correspond to TSVD and CHACI compression, respectively.}
    \label{fig:energy-spin-errors-16-16}
\end{figure}

Figures~\ref{fig:error-gap-16-16} and \ref{fig:energy-spin-errors-16-16} demonstrate that the difference in performance between CHACI and TSVD increases when the active space size is increased from 14-14 to 16-16. In Figure~\ref{fig:error-gap-16-16}, it can be seen that errors in the singlet-triplet gap are at or below 0.1 eV for all cases, when CHACI is employed, even when only 59 kdoubles are stored (compared to 165,637 kdoubles for dense storage). Errors decrease rapidly with additional storage. In contrast, TSVD errors are above 0.2 eV for all cases. Similarly large differences in performance are observed for errors in absolute energy and spin in Figure~\ref{fig:energy-spin-errors-16-16}. As in the 14-14 case, CHACI does a much better job of maintaining the spin symmetry of the wave function than TSVD.

\subsection{Extrapolation of Performance to Larger Active Spaces}
\label{subsec:extrapolation-of-performance}
\begin{figure}[htbp]
    \centering
    \includegraphics[width=5.89in, height=5.89in]{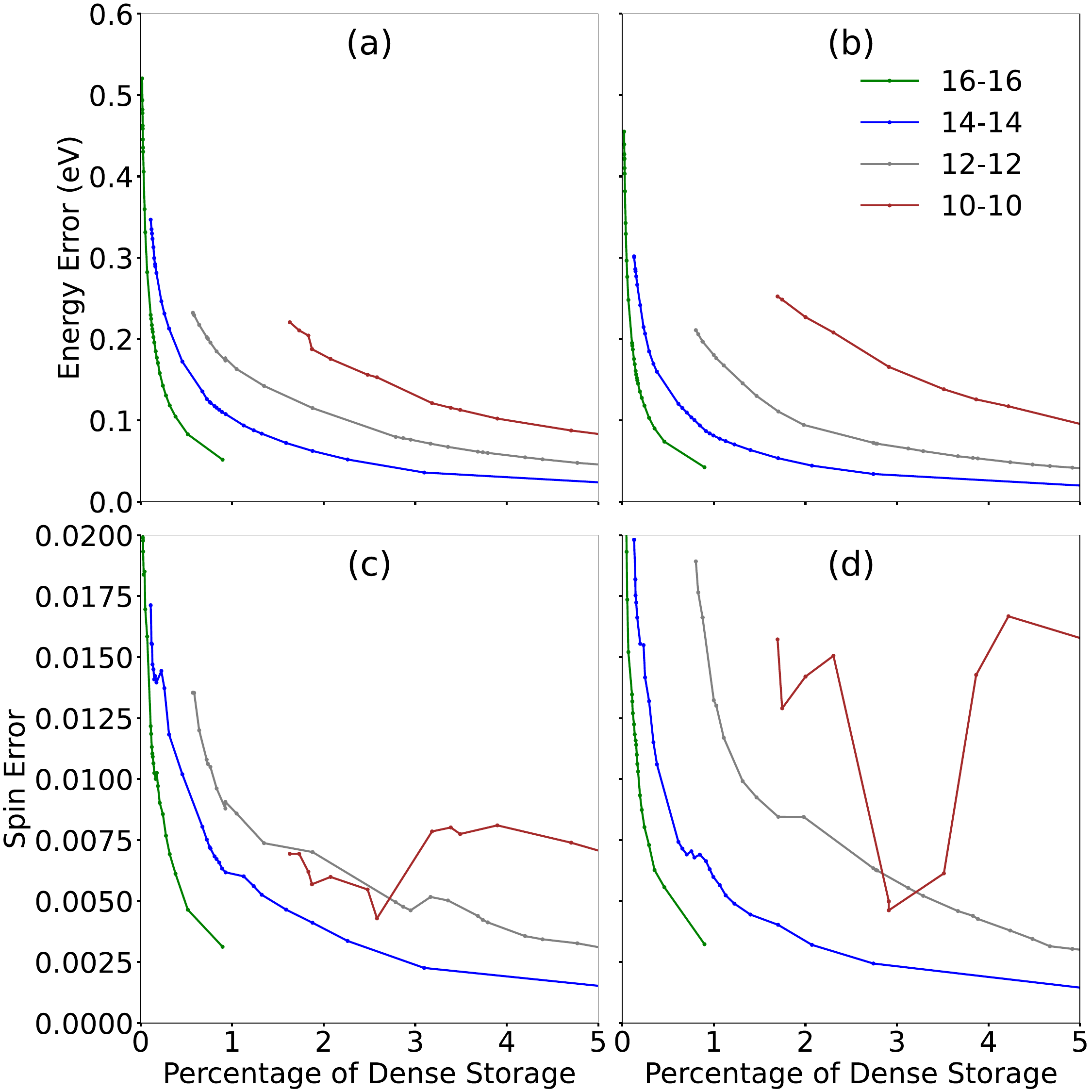}
    \caption{The absolute energy (a and b) and spin errors (c and d) as a function of the percentage of dense storage for 12-acene with 10-10, 12-12, 14-14, and 16-16 active spaces. Panels (a) and (c) correspond to the singlet wave function, while panels (b) and (d) correspond to the triplet spin wave function.}
    \label{fig:errors-vs-compression-ratio}
\end{figure}

Ultimately, our goal is not to compute dense wave functions for subsequent compression. Our goal is to solve for large CI wave functions using a hierarchically compressed basis. Thus, in this section, we consider the behavior of CHACI compression as a function of active space size. In Figure~\ref{fig:errors-vs-compression-ratio}, we consider several active spaces, comparing the convergence of several measures of the accuracy as a function of the percentage of dense storage used (the compression ratio). We find that as the size of the active space increases, the accuracies of both absolute energy and spin converge faster with increasing compression ratio.

\begin{figure}[htbp]
    \centering
    \includegraphics[width=4.55in, height=4.45in]{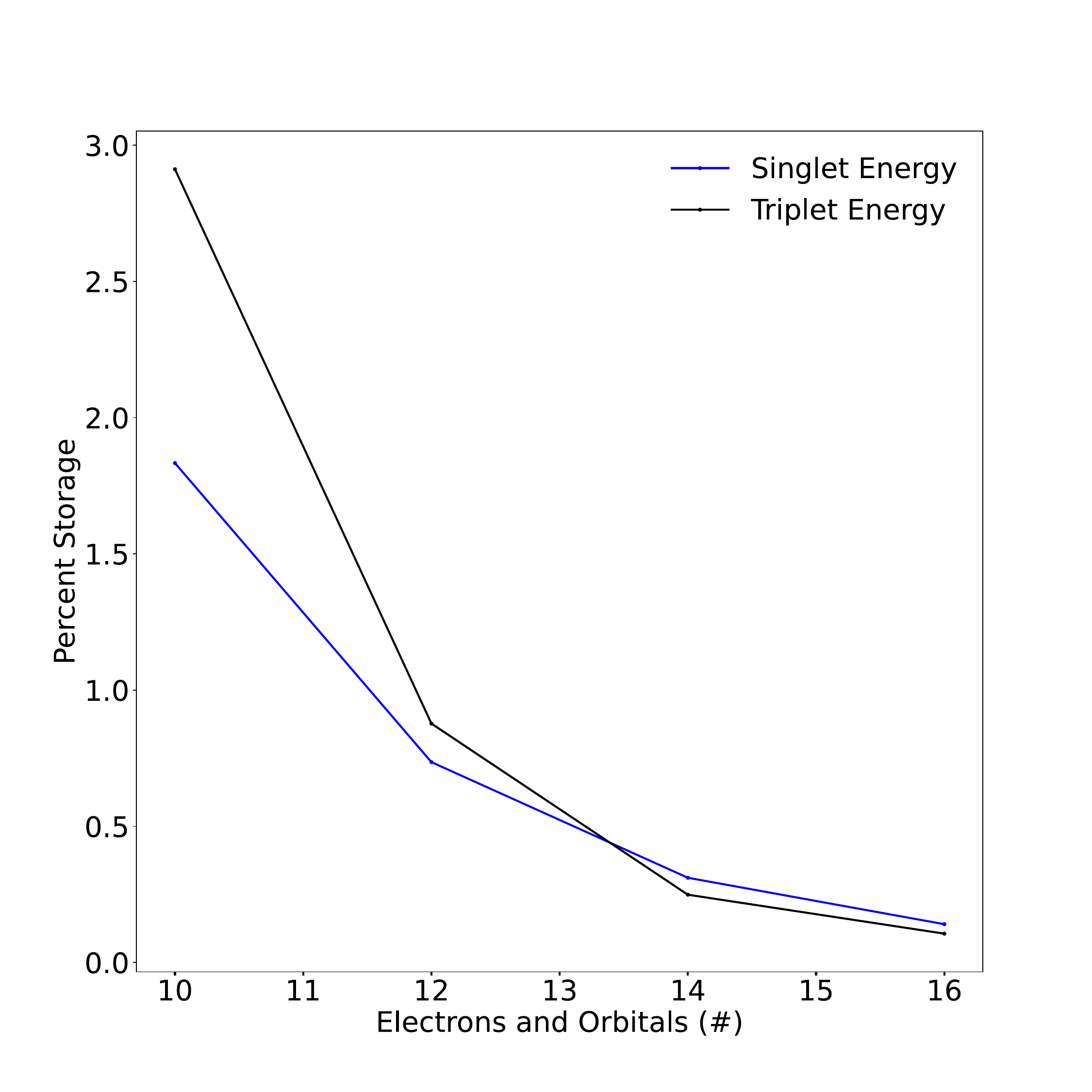}
    \caption{The storage ratio required to achieve \(<0.2\) eV accuracy in absolute energies as a function of the active space size. Blue and black lines correspond to the singlet and triplet wave functions, respectively.}
    \label{fig:compression-ratio-extrapolation}
\end{figure}

To quantify this convergence behavior, we plot the compression ratio at which absolute energies of 0.2 eV accuracy are achieved as a function of the number of active orbitals/electrons in Figure~\ref{fig:compression-ratio-extrapolation}. Both the singlet and triplet compression ratios converge quickly with increasing active space. Of the two, the triplet energy converges more slowly, thus we fit it to an exponential in order to extrapolate to larger active spaces. We find that the required compression ratio decays proportional to
\begin{equation}
f(N_\text{MO}) \propto e^{-0.561N_\text{MO}}.
\end{equation}
Extrapolating to larger active spaces, we estimate that a 24-24 active space could converge to 0.2 eV accuracy at a storage cost of 77,370 kdoubles, which is less than the cost of the dense storage of a 16-16 active space (165,637 kdoubles). Though the convergence behavior is likely to be system-dependent, this result certainly encourages further study.

\subsection{Effect of Using Optimal Rank on Compression}
\label{subsec:effect-of-dynamic-rank-on-compression}
\begin{figure}[htbp]
    \centering
    \includegraphics[width=6.41in, height=4.49in]{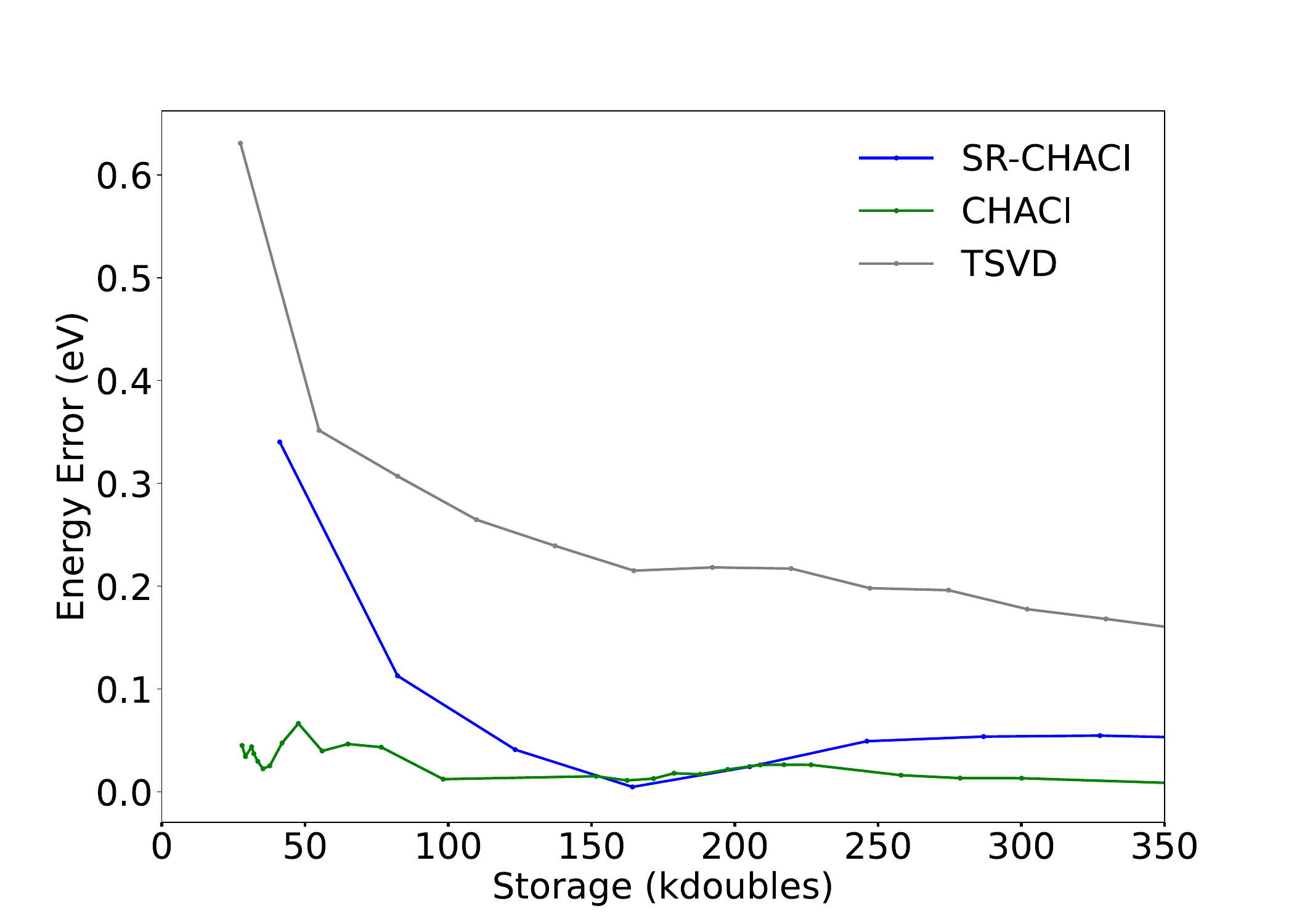}
    \caption{The SR-CHACI error (blue) in the singlet-triplet gap of 12-acene as a function of total storage, computed with a 14-14 active space. The TSVD and CHACI errors (gray and green, respectively) are shown for comparison.}
    \label{fig:SR-CHACI-gap-error}
\end{figure}

In order to assess the necessity of the optimal rank procedure, we compare SR-CHACI (which uses a static rank) to CHACI and TSVD. Figure~\ref{fig:SR-CHACI-gap-error} presents the error in the singlet-triplet gap as a function of the total storage for the 14-14 active space. Excepting one fortunate point at 160 kdoubles, the accuracy of SR-CHACI is significantly worse than that of CHACI, but still better than TSVD. However, considering the error in the absolute energies of the singlet and triplet states separately (Figure 11a and b, respectively), it is clear that this is due to error cancelation. Errors in the absolute energy of the singlet state derived from the SR-CHACI wave function are similar to those of TSVD, and much greater than those of CHACI. Further, errors in the SR-CHACI triplet state are slightly larger than TSVD.

\begin{figure}[htbp]
    \centering
    \includegraphics[width=6.50in, height=6.50in]{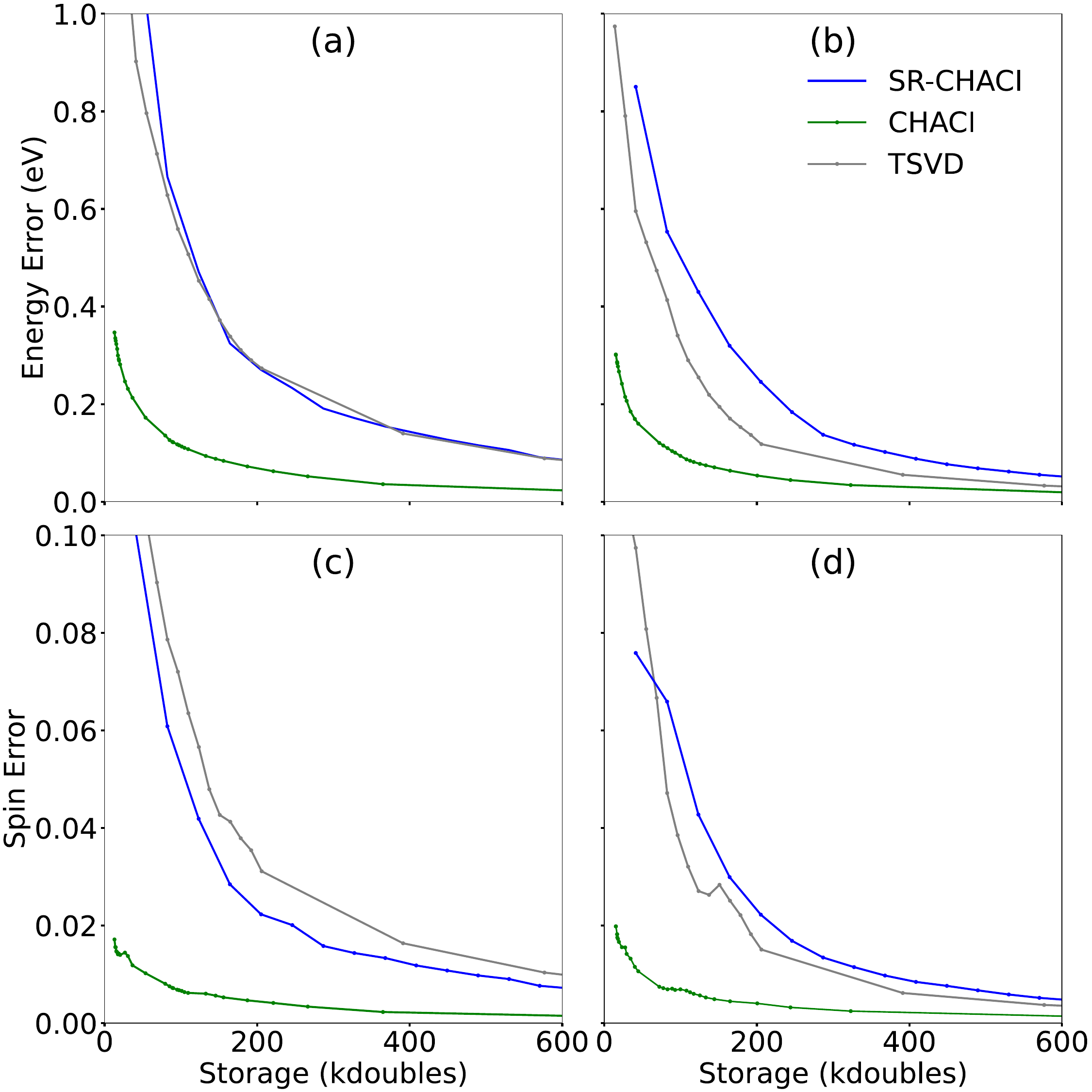}
    \caption{The SR-CHACI errors (blue) in the absolute energy (a and b) and spin (c and d) as a function of the storage for 12-acene (14-14 active space). Panels (a) and (c) correspond to the singlet wave function, while (b) and (d) correspond to the triplet wave function. The TSVD (gray) and CHACI (green) errors are shown for comparison.}
    \label{fig:SR-CHACI-energy-spin-errors}
\end{figure}

Analysis of spin errors tells a similar story. CHACI is much superior to SR-CHACI at reproducing the spin of the original wave function, and SR-CHACI has similar (and sometimes larger) spin errors compared to TSVD. Taken together, we conclude that optimal rank is an essential component of the CHACI algorithm.

\subsection{Effect of Sorting on Compression}
\label{subsec:effect-of-sorting-on-compression}
\begin{figure}[htbp]
    \centering
    \includegraphics[width=4.89in, height=4.69in]{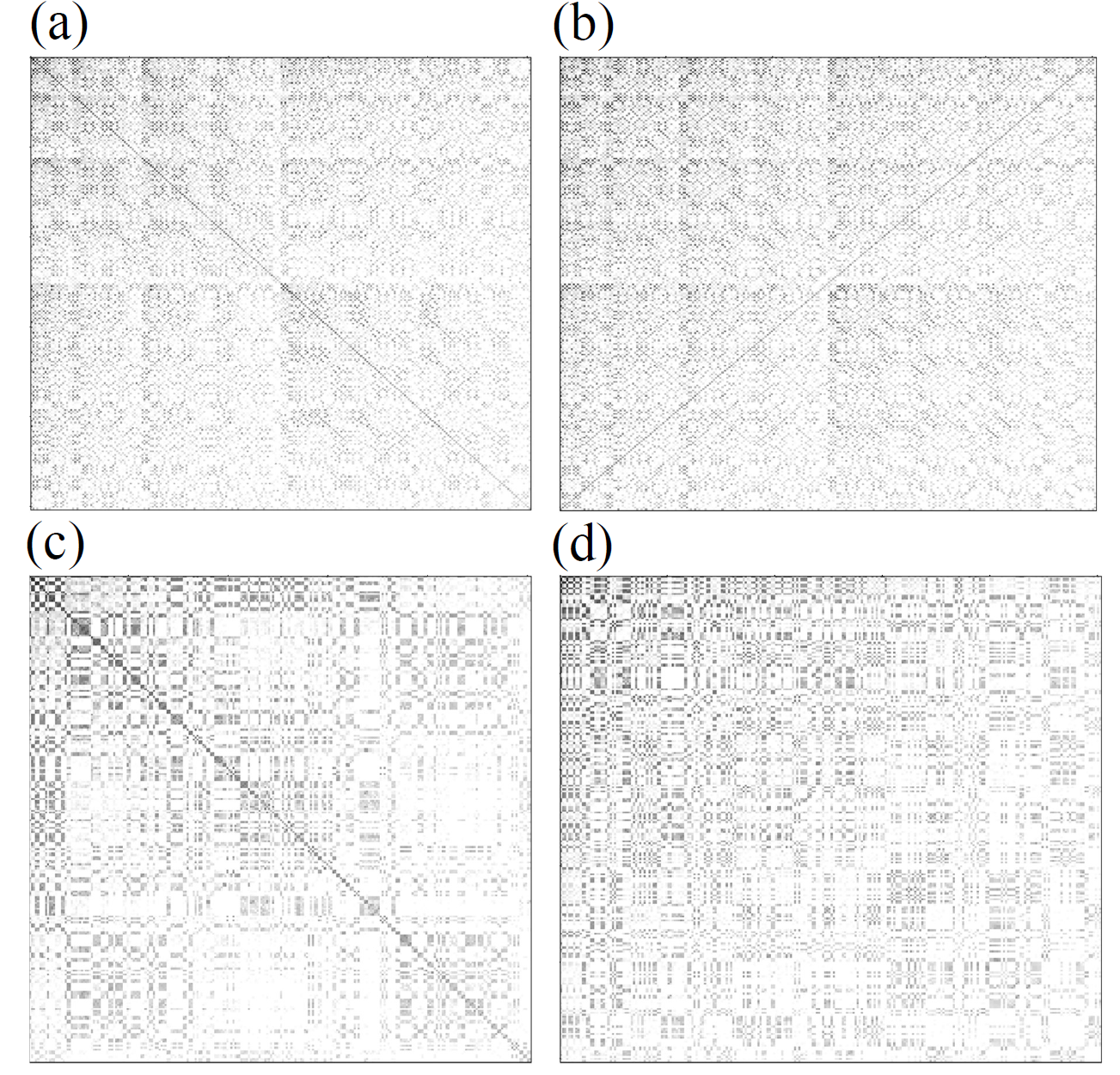}
    \caption{A heat map representation of the CI vector of 12-acene, computed with a 10-10 active space. To form the heat map, we take the logarithm (base 10) of the absolute value of the singlet and triplet CI vector coefficients. The color scale (white to black) ranges from \(10^{-6}\) to 1. Panels (a) and (b) correspond to the unsorted CI vector (with strings indexed according to Duch\cite{Duch1986}), while panels (c) and (d) are reordered according to the CHACI algorithm. Panels (a) and (c) show the singlet wave function, while panels (b) and (d) show the triplet wave function.}
    \label{fig:CI-vector-heatmap}
\end{figure}

Here we assess the necessity of another feature of the CHACI compression algorithm: the sorting of rows/columns of the \(\mathbf{C}\) matrix prior to compression. To this end, we compare U-CHACI, in which the rows/columns remain unsorted, to CHACI and TSVD. Figure~\ref{fig:CI-vector-heatmap} presents a heat-map of the uncompressed \(\mathbf{C}\) matrix of the singlet (panels a and c) and triplet (b and d) wave functions before (a and b) and after (c and d) sorting. Note that sorting concentrates larger elements into the upper-left corner of the matrix.

\begin{figure}[htbp]
    \centering
    \includegraphics[width=4.59in, height=3.22in]{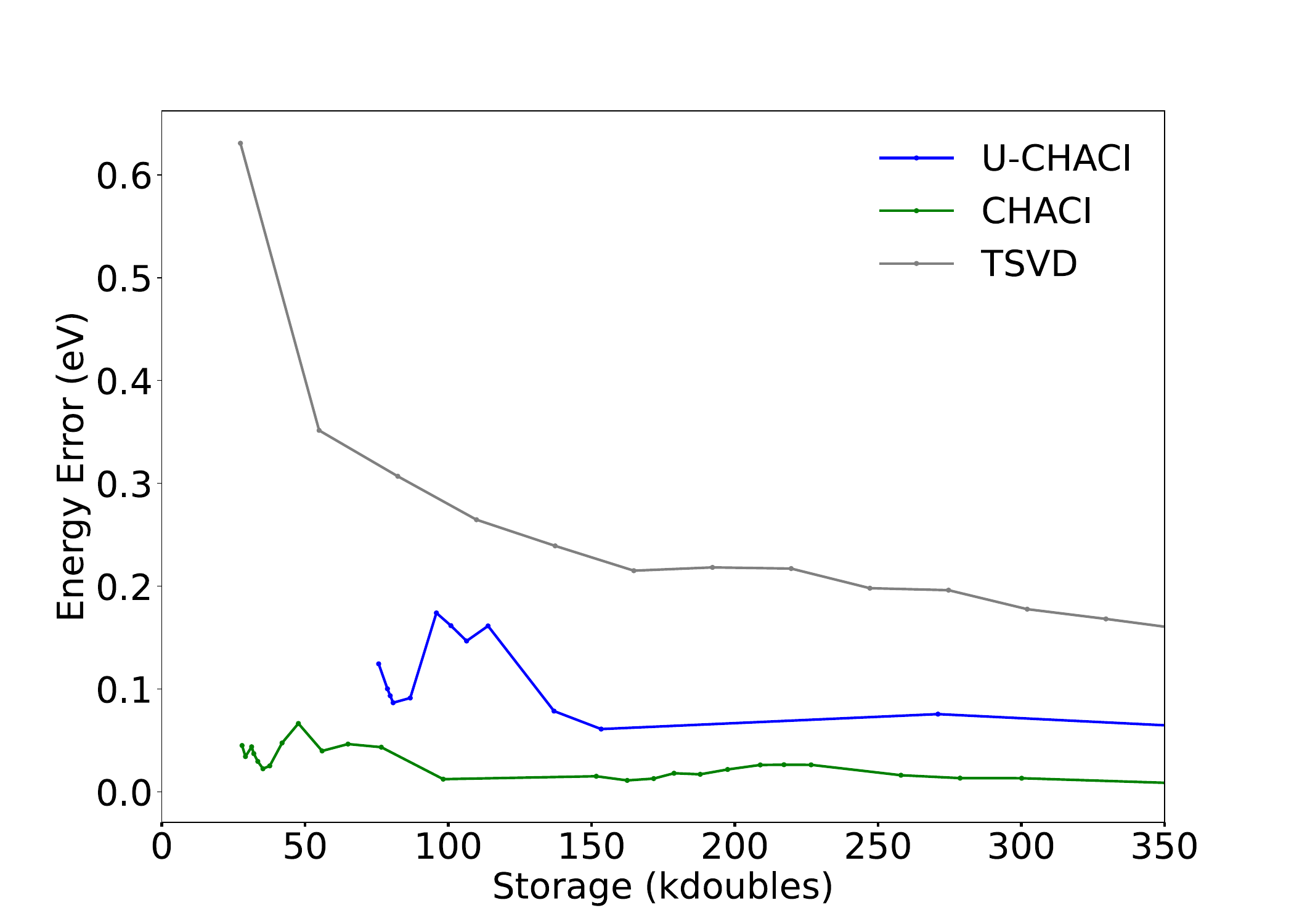}
    \caption{The U-CHACI error (blue) in the singlet-triplet gap of 12-acene as a function of total storage, computed with a 14-14 active space. The TSVD and CHACI errors (gray and green, respectively) are shown for comparison.}
    \label{fig:U-CHACI-gap-error}
\end{figure}

Figure~\ref{fig:U-CHACI-gap-error} compares the U-CHACI singlet-triplet errors to those of CHACI and TSVD. Though U-CHACI appears to be more accurate for predicting the relative energy than TSVD, it remains inferior to CHACI. Considering the errors in the absolute singlet and triplet energies (Figure~\ref{fig:SR-CHACI-energy-spin-errors}a and b), we see that U-CHACI errors are on the order of the same size as those of TSVD, and considerably larger than those of CHACI. That being said, U-CHACI is solidly between CHACI and TSVD in its ability to accurately describe the total spin angular momentum (Figure~\ref{fig:SR-CHACI-energy-spin-errors}c and d).

\begin{figure}[htbp]
    \centering
    \includegraphics[width=4.26in, height=4.26in]{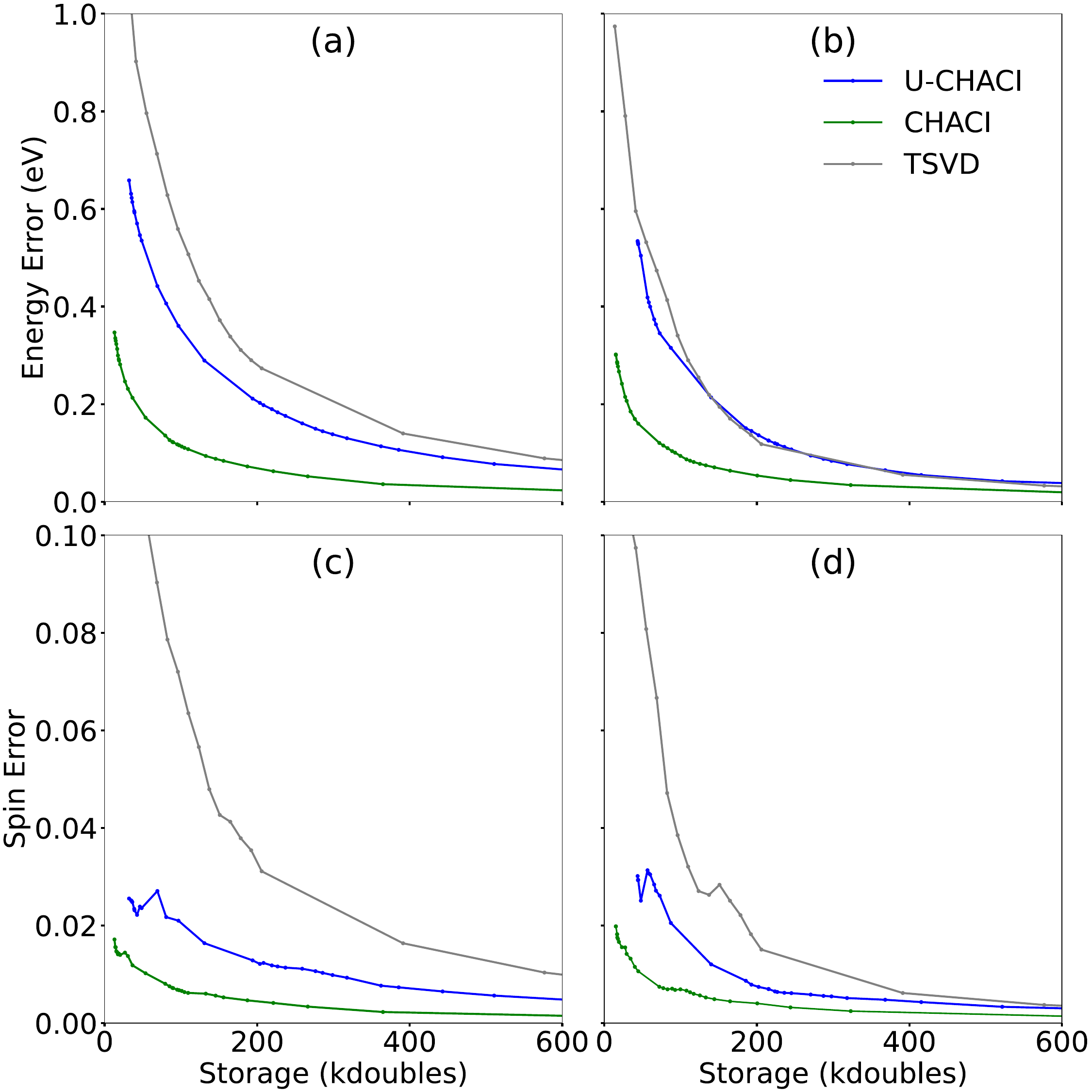}
    \caption{The U-CHACI errors (blue) in the absolute energy (a and b) and spin (c and d) errors as a function of the storage for 12-acene (14-14 active space). Panels (a) and (c) correspond to the singlet wave function, while (b) and (d) correspond to the triplet wave function. The TSVD (gray) and CHACI (green) errors are shown for comparison.}
    \label{fig:U-CHACI-energy-spin-errors}
\end{figure}

Taking this data together, we conclude that sorting is an essential component of the CHACI algorithm. However, our ultimate goal is not to compute the full wave function and subsequently compress it, and the type of \textit{a posteriori} sorting that we use in CHACI would not be possible if we were to directly solve for the wave function in compressed form. But given that the Duch ordering of spin strings does not allow for efficient compression, the determination of an \textit{a priori} scheme by which strings may be ordered for efficient compression remains an important open question.

\subsection{Effect of Upper Quadrant vs H-matrix Compression}
\label{subsec:effect-of-upper-quadrant}
\begin{figure}[htbp]
    \centering
    \includegraphics[width=4.79in, height=4.79in]{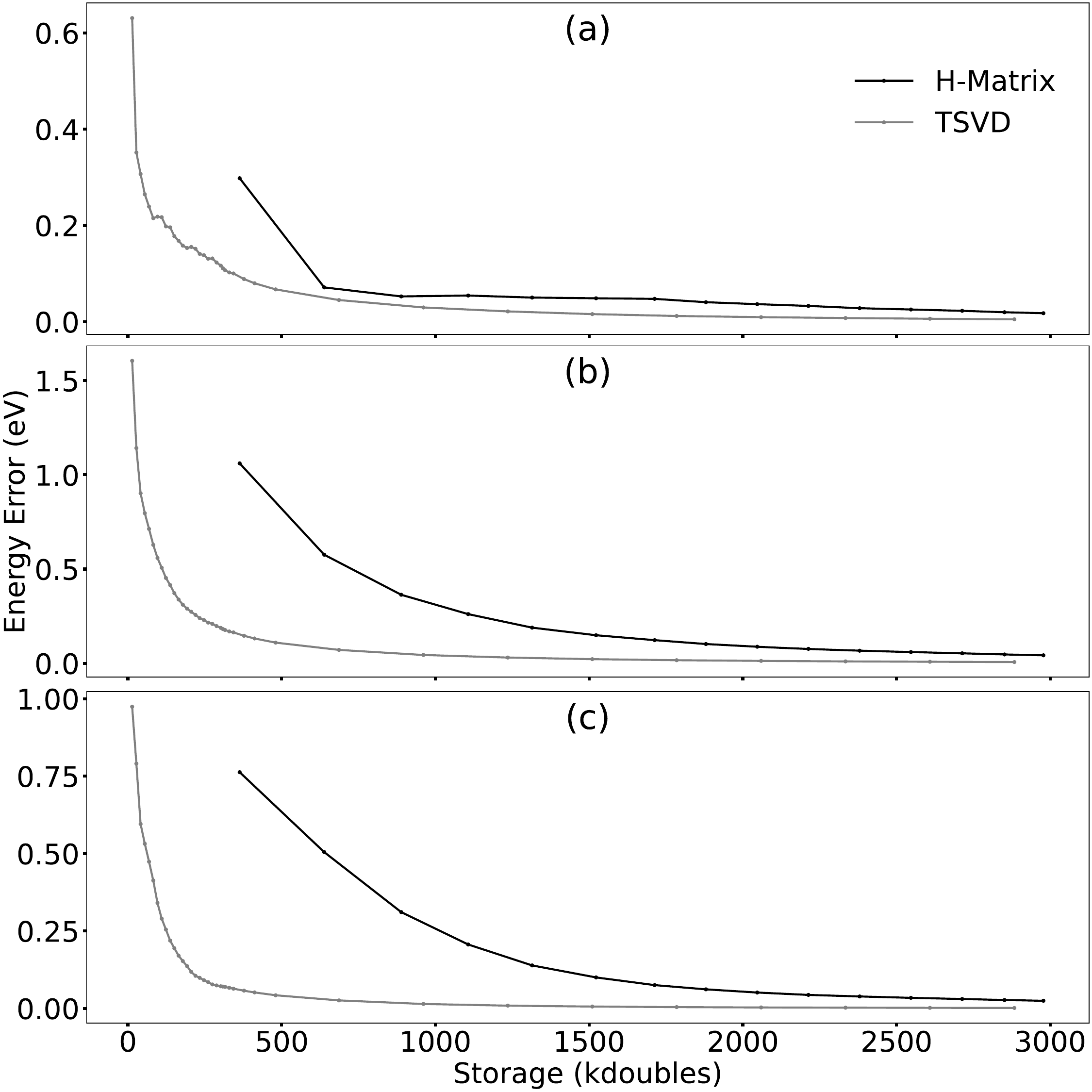}
    \caption{Panels (a), (b), and (c) show the errors in the singlet-triplet gap, singlet, and triplet energies, respectively, of 12-acene (14-14 active space) as a function of the required storage for H-matrix and TSVD compression (black and gray).}
    \label{fig:H-matrix-vs-SVD}
\end{figure}

Lastly, we consider compression of \(\mathbf{C}\) into the H-matrix format,\cite{Hackbusch1999} which is designed to leverage the diagonally dominant nature of the matrix, in contrast to the CH blocking scheme used in CHACI. Figure~\ref{fig:H-matrix-vs-SVD} shows the error in the singlet-triplet gap, singlet absolute energy, and triplet absolute energy of the 14-14 wave function compressed into H-matrix format. H-matrix compression is inferior to TSVD, thus we conclude that the CH blocking scheme is an essential component of the CHACI scheme.

\section{Conclusions}
\label{sec:conclusions}
Here we have presented an exploratory study in which we employ hierarchical matrices to perform lossy compression of CASCI wave functions of a strongly correlated system. We have demonstrated superior compression compared to a global low-rank approximation of the wave function, and we have demonstrated that CH blocking, optimal rank, and row/column sorting are essential features of the algorithm.

These promising results suggest several avenues for future study. The development of an algorithm to directly solve for the wave function in the CHACI representation would allow the accurate solution of strongly correlated wave functions with active spaces well beyond those that are accessible via brute-force CASCI implementations. That hierarchical matrices are well suited for implementation on large-scale distributed computing systems is a noteworthy advantage of this approach. Yet a direct solution will only be possible after an effective \textit{a priori} row/column sorting algorithm is developed, and the development of such an algorithm will require further analysis of the uncompressed wave function. 

\section*{Acknowledgements}

This work was supported by a seed grant from the Institute for Advanced Computational Science. KOB, AT, and BGL are grateful for support from the National Science Foundation under grant CHE-1954519 and for start-up funding from Stony Brook University.

\appendix
\renewcommand{\thesubsection}{\thesection.\arabic{subsection}}

\section{Near-Optimal CHACI Compression}
\label{sec:near-optimal-chaci-compression}
As described above, the CHACI algorithm introduces a strategy for
performing recursive partitioning. Its key idea is to approximate the CI vector via truncated SVD (TSVD) of a collection of small block matrices constructed hierarchically.  The algorithm
leverages the fact that many CI vectors, even of relatively strongly correlated systems, have data that predominantly reside in the matrix's upper-left corner.
Here we investigate this technique and propose an algorithm
that can achieve near-optimal storage efficiency within this framework.
For quantifying approximation accuracy, in this section, we employ the
squared Frobenius norm. Specifically, if \(\mathbf{A}\) represents
the original CI vector and \(\mathbf{A}'\) symbolizes its CHACI
approximation, then the error is quantified by
\(\|\mathbf{A} - \mathbf{A}'\|_F^2\).
Above, the numerical results consider the
energy and spin errors as additional metrics for error evaluation.

\subsection{Information Density}
To derive the algorithm, we first introduce the concept of
\emph{information density}, which quantifies the amount of information
contained in a singular value relative to the storage cost. Consider a matrix $\mathbf{A} \in \mathbb{R}^{m \times n}$ with its singular value decomposition (SVD) given by
\[
\mathbf{A} = \mathbf{U}\mathbf{\Sigma}\mathbf{V}^{\dagger} = [\mathbf{u}_1, \ldots, \mathbf{u}_m]\mathbf{\Sigma}
[\mathbf{v}_1, \ldots, \mathbf{v}_n]^T,
\]
where $\mathbf{u}_i \in \mathbb{R}^m$, $\mathbf{v}_i \in \mathbb{R}^n$, and $\mathbf{\Sigma}$ is an $m\times n$ diagonal matrix composed of $\sigma_i$ along its diagonal. Without loss of generality, let us assume $m\geq n$. Here, $\|\mathbf{A}\|_F^2 = \sum_{i=1}^{n} \sigma_i^2$ represents the total information contained in $\mathbf{A}$. %Storing only the first singular-value triplet $(\mathbf{u}_1, \sigma_1, \mathbf{v}_1)$ retains $\sigma_1^2$ of this information, losing $\sum_{i=2}^n \sigma_i^2$. Any approximation $\mathbf{B} = \mathbf{u}_1 x \mathbf{v}_1^T$ can only achieve a minimum error of $\sum_{i=2}^n \sigma_i^2$ in the squared Frobenius norm.
Storing the first $k$ singular-value triplets would capture $\sum_{i=1}^k \sigma_i^2$ information and lose the remaining $\sum_{i=k+1}^n \sigma_i^2$. This is the best approximation of $\mathbf{A}$ in Frobenius norm, as given by the TSVD. To quantify the efficiency of storing each singular-value triplet $(\mathbf{u}_i, \sigma_i, \mathbf{v}_i)$, we define the \emph{information density} $\rho_i$ as follows:
\begin{definition}
The \emph{information density} $\rho_i$ of the singular-value triplet $(\mathbf{u}_i, \sigma_i, \mathbf{v}_i)$ is the ratio of the information content $\sigma_i^2$ to the storage cost $m+n+1$, i.e.,
\[
\rho_i = \frac{\sigma_i^2}{m+n+1}.
\]
\end{definition}

This measure is particularly useful when comparing matrices of different sizes. For instance, a $4 \times 4$ matrix with a singular value of $5$ has a higher information density of $25/9$ compared to a $8 \times 8$ matrix with a singular value of $6$, which has an information density of $36/17$, less than the $4\times 4$ matrix. This indicates why smaller matrices with higher information density are often preferred for storage when aiming to minimize memory usage while maximizing retained information.

%the ranks of blocks in the CHACI partitioning to minimize memory costs while adhering to a specified accuracy threshold or memory budget. 

\renewcommand{\algorithmicrequire}{\textbf{Input:}}
\renewcommand{\algorithmicensure}{\textbf{Output:}}
\begin{algorithm}[H]
\caption{Near-Optimal CHACI}
\label{alg:optimal-chaci}
\begin{algorithmic}[1]
\Require The sorted CI vector, number of partitioning levels $p$, and information density threshold $\rho$
\Ensure The CHACI format and storage information for each block
\State Split the CI vector into a corner hierarchical structure with $3p+1$ blocks
\State Store the upper-left block in dense format
\For{each of the remaining $3p$ blocks}
\State Let $n_{\text{row}}$ and $n_{\text{col}}$ be the numbers of rows and columns of the block
\State $k\leftarrow 0$, $k_{\text{max}}\leftarrow 1$, and perform rank-$1$ TSVD of the block
\While{$k<\min\{n_{\text{row}},n_{\text{col}}\}/2$}
\If{$k+1>k_{\text{max}}$}
\State \label{alg:line:double-rank} $k_{\text{max}}\leftarrow \min\{2k_{\text{max}}, n_{\text{row}}/2, n_{\text{col}}/2\}$
\State Update the TSVD of this block to rank $k_{\text{max}}$
\EndIf
\If{$\sigma_{k+1}^2\leq \rho(n_{\text{row}}+n_{\text{col}}+1)$}
\State \textbf{break}
\EndIf
\State $k\leftarrow k+1$
\EndWhile
\If{$k=0$}
\State Do not store the block
\ElsIf{$k=\min\{n_{\text{row}},n_{\text{col}}\}/2$}
\State Store the block in dense format
\Else
\State Store the block in TSVD format with rank $k$
\EndIf
\EndFor

\end{algorithmic}
\end{algorithm}

\subsection{Near-Optimal CHACI Algorithm}

Based on the concept of information density, we propose our CHACI algorithm, which strategically chooses between dense and sparse formats for each block. The ``near-optimal'' descriptor emphasizes that among all CHACI configurations, this algorithm seeks to minimize memory usage within a small margin of the optimal solution that achieves the same accuracy. 
Algorithm~\ref{alg:optimal-chaci} outlines this near-optimal CHACI algorithm. This algorithm takes the number of levels $p$ as an input, which should be approximately logarithmic in $M$. Based on our experiments, we choose $p$ to be the smallest integer that satisfies $k_0 2^{p} \geq M$, where $k_0=6$, which can be computed recursively. The algorithm includes a mechanism to switch to dense storage for small blocks. It also takes a threshold for the information density as input, and varying this threshold would result in different storage requirement, as we did in our numerical tests in Section~\ref{sec:results-and-discussion}. It is possible to change the input to be a memory budget instead, and in this setting CHACI possesses some optimality properties, as we analyze next.

%It terminates when the desired accuracy or memory budget is reached. The output includes the CHACI format, storage information for each block, and the Frobenius-norm error.

%It is important to note that in the original CHACI, used for the numerical tests, the number of split levels is computed dynamically and there is no dense part, corresponding to the "else" line in the algorithm.

\subsection{Analysis of Near Optimality}
\label{subsec:analysis-near-opttimality}
%\subsubsection{Sorting Singular-Value Triplets Based on Information Density}

In Section~\ref{sec:results-and-discussion}, we showed numerically that CHACI uses less memory than TSVD. To help understand why this is the case, we present a theoretical analysis of the CHACI algorithm. For simplicity and clarity, we make some simplifying assumptions in our analysis. First, we assume that each block is square, which is the case if the input matrix is a square matrix and its number of rows (and of columns) is a power of two. Second, we change the for and while loops in the algorithm to be in descending order of the information density among all the blocks. This modification requires maintaining a priority queue of the singular triplets of all the blocks, but it does not significantly affect the accuracy and performance compared to using the for-while loops in Algorithm~\ref{alg:optimal-chaci}.

%In the same block, the information density of the larger singular value is always larger than the smaller singular value. Thus, in each step, we only need to use the largest remaining singular value of each block to compare and pick the one with the largest information density. As a result, we can use TSVD to only compute a few largest singular values. If the rank to be stored exceeds the rank of the pre-computed TSVD, we update the TSVD and continue the process.

% Building on the concept of information density introduced earlier, the sorting mechanism within the CHACI framework optimizes storage by evaluating and ordering singular-value triplets based on their information densities. This method is essential for reducing memory usage while maintaining a predefined level of accuracy.

Consider a CHACI-partitioned matrix with $p$ levels, resulting in $3p + 1$ blocks, each represented by an SVD. Denote these blocks as $\mathbf{B}_1$ to $\mathbf{B}_{3p+1}$, with each block $\mathbf{B}_i$ having a size $m_i$. Each singular value $\sigma_{i,j}$ from block $\mathbf{B}_i$ contributes to the overall matrix's Frobenius norm squared, and hence
\[ \|\mathbf{A}\|_F^2 = \sum_{i=1}^{3p+1} \sum_{j=1}^{m_i} \sigma_{i,j}^2. \]
Storing all singular-value triplets provides complete information, but practical constraints often require storage optimization. Therefore, we compute the information density for each block and prioritize those with the highest values to minimize storage while maintaining accuracy. Let
$s_i$ represent the storage of the singular-value triplets and $\rho_i$ the information density. The following lemma motivates us to minimize storage by sorting the singular-value triplets based on $\rho_i$:
\begin{lemma}
\label{lem:chaci-sorting}
Given a collection of positive-number pairs $\{(s_i, \rho_i)\mid 1\leq i\leq n\}$, where $\rho_i$ is sorted in descending order, for an arbitrary selection of $k$ pairs $(s_{j_1}, \rho_{j_1}), \ldots, (s_{j_k}, \rho_{j_k})$, let $r$ be a number between $1$ and $n$ such that
\[ \sum_{i=1}^{r-1} s_i < \sum_{i=1}^k s_{j_i} \leq \sum_{i=1}^r s_i, \]
it then follows that
\[ \sum_{i=1}^r s_i \rho_i \geq \sum_{i=1}^k s_{j_i} \rho_{j_i}. \]
\end{lemma}
\begin{proof}
Let $S=\sum_{i=1}^r(s_i)$. Define a piecewise function $f(x)$ over $0 < x \leq S$ as
\[
f(x) = \begin{cases}
\rho_1 & \text{if } 0 < x \leq s_1, \\
\rho_2 & \text{if } s_1 < x \leq s_1 + s_2, \\
\vdots \\
\rho_r & \text{if } \sum_{i=1}^{r-1} s_i < x \leq \sum_{i=1}^{r} s_i, \\
0 & \text{otherwise}.
\end{cases}
\]
Then,
\[ \sum_{i=1}^r s_i \rho_i = \int_0^{S} f(x) \, dx. \]
Sorting the $\rho_{j_i}$ in descending order, or equivalently sorting $j_i$ in ascending order, we redefine the pairs as $(s_{J_i}, \rho_{J_i})$, where $J_i$ is the $i$th smallest element in the set $\{j_i\}$.
Define $\hat{f}(x)$ similarly as $f(x)$ but with $\rho_{J_i}$ in place of $\rho_i$, i.e.,
\[ \hat{f}(x) = \begin{cases}
\rho_{J_q} & \text{if } \sum_{i=1}^{q-1} s_{J_i} < x \leq \sum_{i=1}^{q} s_{J_i} \text{ for } 1\leq q \leq m,\\
0 & \text{otherwise}.
\end{cases} \]
Since $\rho_i$ is sorted in descending order, we have $f(x) \geq \hat{f}(x)$ for all $x \geq 0$. Therefore,
\[ \sum_{i=1}^r s_i \rho_i - \sum_{i=1}^{m} s_{j_i} \rho_{j_i} = \int_0^S (f(x) - \hat{f}(x)) \, dx \geq 0. \]
\end{proof}

Lemma~\ref{lem:chaci-sorting} implies that by sorting the singular-value triplets based on their information densities, we can maximize the information retained for a given storage cost. This result suggests a greedy algorithm that selects the singular values with the highest information densities first, thereby minimizing storage while maintaining the desired accuracy. However, such a greedy strategy may not always yield the optimal solution for a given accuracy requirement. The following theorem provides a bound on the storage difference between the greedy algorithm and the optimal solution.

\begin{theorem}
\label{thm:storage}
Given a collection of positive-number pairs $\{(s_i, \rho_i)\mid 1\leq i\leq n\}$ as in Lemma~\ref{lem:chaci-sorting} and $t\in(0, \sum_i^n s_i]$, for an arbitrary selection of $k\leq n$ pairs $(s_{j_1}, \rho_{j_1}), \ldots, (s_{j_k}, \rho_{j_k})$ such that $\sum_{i=1}^k s_{j_i} \rho_{j_i} \geq t$, let $r$ be the number of triplets such that
\[ \sum_{i=1}^{r-1} s_i\rho_i < t \leq \sum_{i=1}^r s_i\rho_i, \]
it then follows that
\[
\sum_{i=1}^r s_i - \sum_{i=1}^{k} s_{j_i} < N+2,
\]
where $N$ is the number of rows in the matrix. In other words, the storage difference between the greedy algorithm and the optimal CHACI solution is bounded by a number no more than $N+2$.
\end{theorem}
\begin{proof}
Assume that the selected pairs $(s_{j_i}, \rho_{j_i})$ satisfy that the storage
\[
\sum_i s_{j_i} \leq \sum_{i=1}^{r-1} s_i.
\]
Then by Lemma~\ref{lem:chaci-sorting},
\[
\sum_i s_{j_i} \rho_{j_i} \leq \sum_{i=1}^{r-1} s_i \rho_i < t,
\]
which conflicts with the fact that this selection satisfies the requirement
\[
\sum_i s_{j_i} \cdot \rho_{j_i} \geq t.
\]
Thus we have
\[
\sum_i s_{j_i} > \sum_{i=1}^{r-1} s_i,
\]
so
\[
\sum_{i=1}^r s_i - \sum_i s_{j_i} < \sum_{i=1}^r s_i - \sum_{i=1}^{r-1} s_i = s_r \leq \max_{i=1}^{r}\{s_i\}.
\]
Since the largest block is of size $\lceil N/2\rceil \times \lceil N/2\rceil$, we have $\max{s_i} \leq N+2$.
\end{proof}

In practice, the last block $s_r$ is even much smaller than $N+2$. Hence, the storage difference between the greedy algorithm and the optimal solution is bounded by a number typically far smaller than $N$. Therefore, CHACI is nearly optimal in its storage requirement. In the following, we further show that CHACI typically requires less memory and is more efficient than applying the TSVD to the global matrix.

% \subsubsection{Further Optimization}

% One assumption in the above lemma is that all the blocks are stored in TSVD format as triplets. However, for small enough blocks, dense storage may be more efficient. Consider a block $\mathbf{B}$ of size $m \times m$ where we already store $k$ singular values. The information density of storing the next singular value $\sigma_{k+1}$ is $\sigma_{k+1}^2/(2m+1)$. In contrast, changing to a dense format results in a storage cost of $m^2$ but retains full information, $\|\mathbf{B}\|_F^2$. In this case, the information density is $\|\mathbf{B}\|_F^2/m^2$. We can further optimize the storage by switching to the dense format if
% $$\frac{\|\mathbf{B}\|_F^2 - \sum_{i=1}^{k} \sigma_i^2}{m^2 - k(2m + 1)} \geq \frac{\sigma_{k+1}^2}{2m+1}.$$
% These considerations set the stage for the algorithm described in the subsequent section, which incorporates these principles to optimize storage across multiple levels of CHACI partitioning. With this further optimization, we deviate from the CHACI framework slightly, and the greedy algorithm under this extended framework can yield smaller storage than the optimal CHACI algorithm, although it may not be as memory efficient as the optimal algorithm within this extended framework.

\subsection{Analysis of Complexity}
\label{subsec:implementation-and-complexity}

The CHACI algorithms assume that the input CI vector is fully represented, requiring \(N^2\) storage. The extra intermediate storage of the algorithm is about the same as the input size. Hence, we focus on the output's storage requirement, a common approach for compression techniques.

For each block \(i\) of size \(m_i \times m_i\), if stored in TSVD format with rank \(k_i\), the storage requirement includes \(2k_im_i\) for the vectors (each vector has \(m_i\) entries) and \(k_i\) for the singular values, totaling \(2k_im_i + k_i\). If a block is stored in a dense format, its requirement is \(m_i^2\). Thus, the total storage for the output is
\[
\text{Total output size} = \sum_{\text{TSVD block } i} (2k_im_i + k_i) + \sum_{\text{dense block } i} m_i^2.
\]
It is worth noting that the dense blocks are typically for the smaller blocks, and the total storage is dominated by the TSVD blocks. Note that
\[
\sum_{i=1}^{3p+1} m_i \leq \sum_{j=1}^{p} \frac{3 N}{2^{j}} +\left\lceil \frac{N}{2^p} \right\rceil\leq 3N+1,
\]
and
\[
\sum_{i=1}^{3p+1} (2m_i + 1) \leq 6N+6p+3=6N+o(N).
\]
For simplicity, assume all the blocks are stored in TSVD format. The output can be bounded by
\[
\text{Total output size} \lesssim 6N\max\{k_i\}
\]
for sufficiently large $N$. This bound is tight if the ranks $\{k_i\}$ are uniform, but it is pessimistic if the ranks are highly non-uniform, as they are in practice for CI vectors.

Now we compare the storage requirement to that of global TSVD.
A global TSVD for the CI vector traditionally stores \(K\) singular values and vectors, resulting in an output storage of \(2KN + K\).
Hence, if the input matrix were a uniform low-rank matrix, CHACI would require more storage than the global TSVD.\footnote{H-matrices would also require more storage than the global TSVD if the input matrix is a uniform low-rank matrix.} However, CHACI is designed to exploit the structure of CI vectors, which are not uniformly low-rank. In our numerical experimentation, the
output size of near-optimal CHACI is typically about 2\% of that of a global TSVD of equal accuracy.

% \subsubsection{Time Complexity}
% \label{subsec:time-complexity}
The time complexity of the near-optimal CHACI algorithm involves computing TSVDs for various blocks. % and managing a priority queue that organizes these blocks by information density.
For each block \(i\) of size \(m_i \times m_i\), a rank-\(k_i\) TSVD can be approximately computed in \(\mathcal{O}(k_im_i^2)\) using iterative methods.
%The priority queue, with a length of \(3p+1\) (where \(p \leq \lceil\log_2 N\rceil\)), manages the blocks, and its operations (insertions and deletions) are \(\mathcal{O}(\log(3p+1))\) per operation. This cost is negligible compared to the quadratic scaling of TSVD computations of the largest block.
An important detail worth emphasizing is Line~\ref{alg:line:double-rank} of Algorithm~\ref{alg:optimal-chaci}. In this step, we would double the TSVD rank if the desired accuracy is not met. Such a doubling process results in a geometric progression of the TSVD computation cost, calculated as follows:
\[
\sum_{j=1}^{\lceil \log_2 k_i \rceil} 2^j m_i^2 \leq 4k_i m_i^2.
\]
Consequently, the total time complexity for all blocks in the CHACI algorithm is
\[
\text{Total time complexity} = \sum_{i=1}^{3p+1} \mathcal{O}(k_i m_i^2),
\]
where \(\max\{m_i\}=N/2\). If the rank increases were linear when recomputing the TSVD, the worst-case time complexity could potentially reach \(\mathcal{O}(k_i^2 m_i^2)\) for each block, potentially increasing the overall time complexity.

Now we compare the time complexity to that of global TSVD.
Performing a global TSVD of an \(N \times N\) matrix with a truncation rank \(K\) incurs a time complexity of \(\mathcal{O}(KN^2)\). Since
\[
\sum_{i=1}^{3p+1} m_i^2 \leq \sum_{j=1}^{\lceil \log N \rceil} 3\frac{N^2}{2^{2j}} \leq N^2.
\]
As long as \( k_i < K \), the near-optimal CHACI algorithm has a lower computational cost than a global TSVD. In practice, \( k_i \ll K \) in the near-optimal CHACI algorithm, as we will demonstrate in our numerical experimentation. Hence, we can expect the near-optimal CHACI algorithm to be significantly faster than a global TSVD for moderately large $N$.

\nocite{}

\bibliography{refs}

\end{document}